\documentclass[11pt,a4paper,reqno]{amsart}%
\usepackage[latin1]{inputenc}
\usepackage{mathrsfs}
\usepackage{dsfont}
\usepackage{hyperref}
\usepackage{amsmath}
\usepackage{amssymb}
\usepackage{amsthm}
\usepackage{amsfonts}
\usepackage{amstext}
\usepackage{amsopn}
\usepackage{amsxtra}
\usepackage{mathrsfs}
\usepackage{esint}
\usepackage{pst-all}
\usepackage{graphicx}
\theoremstyle{plain}
\newtheorem{theorem}{Theorem}[section]
\newtheorem{lemma}[theorem]{Lemma}
\newtheorem{corollary}[theorem]{Corollary}
\newtheorem{proposition}[theorem]{Proposition}

\theoremstyle{definition}

\newtheorem{example}[theorem]{Example}
\theoremstyle{remark}
\newtheorem{remark}[theorem]{Remark}

\numberwithin{equation}{section}

\newcommand{\ii}{\infty}
\newcommand\R{{\ensuremath {\mathbb R} }}

\newcommand\N{{\ensuremath {\mathbb N} }}
\newcommand\Z{{\ensuremath {\mathbb Z} }}

\newcommand\1{{\ensuremath {\mathds 1} }}
\renewcommand\phi{\varphi}

\newcommand{\gS}{\mathfrak{S}}

\newcommand{\wto}{\rightharpoonup}
\newcommand{\cS}{\mathcal{S}}
\newcommand{\cP}{\mathcal{P}}

\newcommand{\cM}{\mathcal{M}}

\newcommand{\cB}{\mathcal{B}}

\newcommand{\cE}{\mathcal{E}}
\newcommand{\cETF}{\mathcal{E}_{\rm TF}}

\newcommand{\cW}{\mathcal{W}}
\newcommand{\cF}{\mathcal{F}}

\newcommand{\eps}{\epsilon}

\newcommand{\by}{\mathbf{y}}

\renewcommand{\epsilon}{\varepsilon}
\newcommand\pscal[1]{{\ensuremath{\left\langle #1 \right\rangle}}}
\newcommand{\norm}[1]{ \left| \! \left| #1 \right| \! \right| }
\newcommand{\tr}{{\rm Tr}\,}

\renewcommand{\geq}{\geqslant}
\renewcommand{\leq}{\leqslant}

\renewcommand{\tilde}{\widetilde}

\newcommand{\nn}{\nonumber}

\title{The semi-classical limit of large fermionic systems}

\author[S. Fournais]{S{\o}ren Fournais}
\address{Department of Mathematics, Aarhus University, Ny Munkegade 118, DK-8000 Aarhus C, Denmark} 
\email{fournais@math.au.dk}

\author[M. Lewin]{Mathieu Lewin}
\address{CNRS \& Universit\'e Paris-Dauphine, CEREMADE (UMR 7534), Place de Lattre de Tassigny, F-75775 Paris Cedex 16, France} 
\email{mathieu.lewin@math.cnrs.fr}

\author[J. P. Solovej]{Jan Philip Solovej}
\address{Department of Mathematics, University of Copenhagen, Universitetsparken 5, DK-2100 Copenhagen {\O}, Denmark} 
\email{solovej@math.ku.dk}

\date{\today}

\begin{document}

\begin{abstract}
We study a system of $N$ fermions in the regime where the intensity of the interaction scales as $1/N$ and with an effective semi-classical parameter $\hbar=N^{-1/d}$ where $d$ is the space dimension. For a large class of interaction potentials and of external electromagnetic fields, we prove the convergence to the Thomas-Fermi minimizers in the limit $N\to\infty$. The limit is expressed using many-particle coherent states and Wigner functions. The method of proof is based on a fermionic de Finetti-Hewitt-Savage theorem in phase space and on a careful analysis of the possible lack of compactness at infinity.
\end{abstract}

\maketitle

\tableofcontents

\section{Introduction and main results}

Large interacting quantum systems are extremely difficult to describe accurately, due to the high complexity of the many-particle Schrödinger equation. It is therefore useful to derive approximate models that are much easier to handle and, at the same time, have a proper physical behavior. In the \emph{mean-field model}, the real interaction between the particles is replaced by a self-consistent effective potential depending on the average density of particles in the system. It is believed that this model becomes a good approximation to the true many-particle problem when the system is sufficiently dense (such that the particles meet very often) and the interactions are weak (for a law of large numbers to hold). The purpose of this paper is to discuss the validity of the mean-field model for fermions in a semi-classical-type limit.

\subsubsection*{The mean-field Hamiltonian for bosons and fermions}
The Hamiltonian that will be the object of our attention describes $N$ spinless quantum particles in $\R^d$, where $d\geq1$ is arbitrary (of course physically $d=1,2,3$). In units such that $m=1/2$, it takes the form
\begin{equation}
H_{\hbar,N}=\sum_{j=1}^N\big(-i\hbar\nabla_{j}+A(x_j)\big)^2+V(x_j)+\frac{1}{N}\sum_{1\leq k<\ell\leq N}w(x_k-x_\ell).
 \label{eq:Hamiltonian_intro}
\end{equation}
For bosons this operator must be restricted to the subspace $\bigotimes_s^N L^2(\R^d)\subset L^2(\R^{dN})$ containing all the functions that are symmetric under exchange of variables, that is,
$$\Psi(x_{\sigma(1)},...,x_{\sigma(N)})=\Psi(x_1,...,x_N)$$
for every permutation $\sigma$. For fermions, it must be restricted to the subspace $\bigwedge^N L^2(\R^d)\subset L^2(\R^{dN})$ containing all the functions that are \emph{anti}symmetric,
$$\Psi(x_{\sigma(1)},...,x_{\sigma(N)})={\rm sgn}(\sigma)\Psi(x_1,...,x_N).$$
In~\eqref{eq:Hamiltonian_intro}, $A$ plays the role of a magnetic vector potential but can have a different physical origin (e.g.~for rotating gases), $V$ is an external potential that traps the particles (possibly only in a finite region of space if $V\to0$ at infinity), and $w$ is the two-particle interaction potential (that could in principle depend on $N$ as well, but that will be kept fixed here). The factor $1/N$ in front of the interaction term is typical of the mean-field scaling and its role is to make the two sums of the same order $N$ in the Hamiltonian. In the limit of large $N$, the expectation is that the particles behave independently. In that case, the law of large numbers tells us that the $j$th particle experiences the mean-field potential
$$\frac1{N}\sum_{k\neq j}w(x_j-x_k)\simeq\int_{\R^d}w(x_j-y)\rho(y)\,dy$$
where $\rho$ is the density of particles in the system. 

The validity of the mean-field approximation can be studied in the time-dependent or in the time-independent case. For the time-dependent equation, one usually assumes that the particles are independent at the initial time and then proves the propagation of chaos (that is, of the independence) for all times in the limit of large $N$. The situation is slightly different for the stationary problem where one has to prove that (say, in the ground state) the particles are independent. This second situation will be the object of our study.

For \textbf{bosons}, the validity of the mean-field approximation has been studied in the limit $N\to\ii$ with $\hbar$ fixed, first in some specific situations~\cite{LieLin-63,BenLie-83,LieThi-84,LieYau-87,Solovej-90,Bach-91,BacLewLieSie-93,FanSpoVer-80,VdBLewPul-88,RagWer-89,Werner-92,Kiessling-12,LieSeiSolYng-05,LieSei-06,Seiringer-11,GreSei-13,SeiYngZag-12}. The most general case has been addressed in a recent series of works~\cite{LewNamRou-14,LewNamRou-15b,LewNamRou-16c,LewNamRou-17b,Rougerie-15} by Nam, Rougerie and the second author of the present article. Under rather general assumptions on $A$, $V$ and $w$, it is possible to prove that the bosonic ground state energy per particle converges to that of the nonlinear Hartree functional
\begin{multline}
\cE^{V,A}_\text{Hartree}(u)=\int_{\R^d}\left(|-i\hbar \nabla u(x)+A(x)u(x)|^2+V(x)|u(x)|^2\right)\,dx\\+\frac{1}{2}\int_{\R^d}\int_{\R^d}w(x-y)|u(x)|^2|u(y)|^2\,dx\,dy.
\label{eq:Hartree}
\end{multline}
In mathematical terms, the bottom of the spectrum of $H_{\hbar,N}$ restricted to the bosonic subspace $\bigotimes_s^N L^2(\R^d)$, divided by $N$, converges to the infimum of $\cE_\text{Hartree}^{V,A}$:
$$\lim_{N\to\ii}\frac{\inf\sigma_{\bigotimes_s^N L^2(\R^d)}(H_{\hbar,N})}{N}=\inf_{\int_{\R^d}|u|^2=1}\cE^{V,A}_\text{Hartree}(u).$$
For the convergence of states and the link with the minimizers of $\cE^{V,A}_\text{Hartree}$, we refer to~\cite{LewNamRou-14} and to the discussion below. 
The link between the time-dependent Schrödinger equation $i\dot\Psi=H_{\hbar,N}\Psi$ and the nonlinear time-dependent Hartree equation 
$$i\frac{\partial}{\partial t}u = \big(-i\hbar\nabla+A\big)^2u+Vu+(|u|^2\ast w)u$$
has stimulated many works as well~\cite{Hepp-74,GinVel-79,Spohn-80,BarGolMau-00,ElgErdSchYau-06,ElgSch-07,AmmNie-08,ErdSchYau-09,FroKnoSch-09,RodSch-09,KnoPic-10,Pickl-11,LewNamSch-15}. The coupled limit $N\to\ii$ and $\hbar\to0$ has been investigated for the time-dependent problem in~\cite{LioPau-93,GraMarPul-03,FroGraSch-07}. 

The purpose of this work is to address the case of \textbf{fermions}. The anti-symmetry (called the \emph{Pauli principle}) implies that two particles cannot be at the same place and this usually makes the kinetic energy grow much faster than $N$, by the Lieb-Thirring inequality~\cite{LieThi-75,LieThi-76,LieSei-09}. More precisely, if $\Psi_N$ is an $N$-particle normalized antisymmetric function with support in a bounded domain $\Omega^N\subset\R^{dN}$, then we have
$$\sum_{j=1}^N\int_{\Omega^N}|\nabla_j\Psi|^2\geq C|\Omega|^{-\frac{2}{d}}N^{1+\frac{2}{d}}.$$
In order to ensure that all the terms in the Hamiltonian are of the same order, it is therefore necessary for fermions to take
$$\boxed{\hbar=\frac{1}{N^{\frac1d}},}$$
which means that the mean-field limit is coupled to a semi-classical limit. We therefore end up with the Hamiltonian
\begin{equation}
\boxed{H_N:=\sum_{j=1}^N\left(\frac{-i\nabla_{j}}{N^{\frac{1}d}}+A(x_j)\right)^2+V(x_j)+\frac{1}{N}\sum_{1\leq k<\ell\leq N}w(x_k-x_\ell),}
\label{eq:Hamiltonian}
\end{equation}
that will be our main object of interest in this work. Its fermionic ground state energy (that is, the bottom of the spectrum in the fermionic subspace $\bigwedge^N L^2(\R^d)$) will be denoted by
\begin{equation}
\boxed{E(N)=\inf\sigma_{\bigwedge^N L^2(\R^d)}(H_N).}
\end{equation}

A given physical Hamiltonian is not necessarily in the form~\eqref{eq:Hamiltonian_intro} but it can sometimes be recast in that form after an appropriate scaling. This is for instance the case for atoms~\cite{LieSim-77b,LieSim-77} (in units of length of the order $Z^{-1/3}$) and non-relativistic fermion stars~\cite{LieThi-84,LieYau-87}, see~\eqref{eq:atoms} and~\eqref{eq:stars} below.

\subsubsection*{Vlasov and Thomas-Fermi theories}
The Hamiltonian $H_N$ has been studied a lot in the time-dependent setting~\cite{NarSew-81,Spohn-81,BarGolGotMau-03,ElgErdSchYau-04,FroKno-11,BenPorSch-14,BenJakPorSafSch-15,BacBrePetPicTza-15,PetPic-14,BenPorSafSch-15}, where its dynamics is known to converge to the Vlasov time-dependent equation in the limit $N\to\ii$. This is the Hamiltonian dynamics associated with the \emph{Vlasov energy}, which is the semi-classical equivalent to the Hartree functional~\eqref{eq:Hartree}:
\begin{multline}
\cE^{V,A}_{\rm Vla}(m)=\frac{1}{(2\pi)^d}\int_{\R^d}\int_{\R^d}|p+A(x)|^2m(x,p)\,dx\,dp+\int_{\R^d}V(x)\rho_m(x)\,dx\\+\frac12\iint_{\R^d\times\R^d}w(x-y)\rho_m(x)\,\rho_m(y)\,dx\,dy.
\label{eq:def_Vlasov}
\end{multline}
Here 
$$\rho_m(x)=\frac{1}{(2\pi)^d}\int_{\R^d}m(x,p)\,dp.$$
and $m(x,p)$ is a probability measure on the phase space $\R^d\times\R^d$ that must satisfy the additional constraint 
\begin{equation}
\boxed{0\leq m(x,p)\leq 1\quad\text{a.e.}}
\label{eq:constraint_m}
\end{equation}
This condition says that one cannot put more than one particle at $x$ with a momentum $p$ and it is inherited from the Pauli principle. 

We will look at the ground state problem, that is, we study the minimization of the Vlasov energy~\eqref{eq:def_Vlasov}. If we minimize in the variable $p$ for any fixed $x$ without the constraint~\eqref{eq:constraint_m}, the ground state measures all have the trivial momentum distribution $\delta_0(p+A(x))$ and the first term disappears. With the fermionic constraint~\eqref{eq:constraint_m}, the Vlasov energy is minimized for measures of the form
\begin{equation}
m_\rho(x,p)=\1\left(|p+A(x)|^2\leq c_{\rm TF}\,\rho(x)^{2/d}\right)
 \label{eq:def_m_rho}
\end{equation}
where $\rho$ now minimizes the \emph{Thomas-Fermi energy}
\begin{multline}
\cE^V_{\rm TF}(\rho):=\cE^{V,A}_{\rm Vla}(m_\rho)=\frac{d}{d+2}c_\text{TF}\int_{\R^d}\rho(x)^{1+\frac{2}{d}}\,dx+\int_{\R^d}V(x)\rho(x)\,dx\\
+\frac12\iint_{\R^d\times\R^d}w(x-y)\rho(x)\,\rho(y)\,dx\,dy
\label{eq:def_TF}
\end{multline}
and 
$$c_\text{TF}=4\pi^2\left(\frac{d}{|S^{d-1}|}\right)^{\tfrac{2}{d}}.$$
The functional is similar to the Hartree energy~\eqref{eq:Hartree} with $\rho=|u|^2$ except that the quantum term $\int_{\R^d}|\nabla u|^2$ has been replaced by the fermionic classical kinetic energy proportional to $\int_{\R^d}\rho^{1+2/d}$ and that the magnetic field has been discarded (but it appears in the formula~\eqref{eq:def_m_rho} of the minimizers on phase space). 

We are interested in the link between the many-particle ground state energy $E(N)$ and that of the Vlasov and Thomas-Fermi energies in the limit $N\to\ii$ with $\hbar=N^{-1/d}$. To this end we introduce the Thomas-Fermi ground state energy
\begin{align*}
e^V_{\rm TF}(\lambda)&:=\inf\left\{\cETF^V(\rho)\ :\ 0\leq \rho\in L^1(\R^d)\cap L^{1+2/d}(\R^d),\ \int_{\R^d}\rho=\lambda\right\}\\
&=\inf\left\{ \cE^{V,A}_{\rm Vlas}(m)\ :\ 0\leq m\leq 1,\ (2\pi)^{-d}\int_{\R^{2d}}m=\lambda\right\},
\end{align*}
where $\lambda=1$ in our case. The equality of the two infima is obtained by inserting the formula~\eqref{eq:def_m_rho} in the Vlasov energy~\eqref{eq:def_Vlasov}.

\subsubsection*{Convergence of $E(N)/N$}
Our first result is that the leading order of $E(N)$ is given by the Thomas-Fermi ground state energy $e^V_{\rm TF}(1)$.

\begin{theorem}[Convergence of the ground state energy]\label{thm:CV_GS_energy}
Assume that $w$ is even, that $w,|A|^2\in L^{1+d/2}(\R^d)+L^\ii_\epsilon(\R^d)$ and that either 
$$V\in  L^{1+d/2}(\R^d)+L^\ii_\epsilon(\R^d)$$ 
or 
$$V_-\in  L^{1+d/2}(\R^d)+L^\ii_\epsilon(\R^d),\qquad V_+\in L^1_{\rm loc}(\R^d),\qquad \lim_{|x|\to\ii}V_+(x)=+\ii.$$
Then we have
\begin{equation}
\boxed{\lim_{N\to\ii}\frac{E(N)}{N}=e^V_{\rm TF}(1).}
\label{eq:CV_GS_energy}
\end{equation}
\end{theorem}

We recall that $f\in L^p(\R^d)+L^\ii_\epsilon(\R^d)$ means that for every $\epsilon>0$, we can write $f=f_p+f_\ii$ with $f_p\in L^p(\R^d)$ and $\norm{f_\ii}_\ii\leq \epsilon$.  The assumptions on the potentials cover most physically interesting systems in this regime. Our result can easily be generalized in many directions but we do not state precisely the corresponding theorems. Our two assumptions on $V$ cover either systems which are locally confined (the assumption $V\in L^{1+d/2}(\R^d)+L^\ii_\epsilon(\R^d)$ essentially means that $V\to0$ at infinity) or confined (when $V_+\to+\ii$ at infinity). Our proof works the same if $V$ has a different limit at infinity or if 
$V=+\ii$ outside of a domain $\Omega$, which corresponds to Dirichlet boundary conditions. We could also weaken the assumptions on $w_+$ but we refrain from doing it in order to simplify the next statement.  It is also possible to use a pseudo-relativistic kinetic energy
\begin{equation}
\sqrt{\left(\frac{-i\nabla_{j}}{N^{\frac{1}d}}+A(x_j)\right)^2+m^2} 
 \label{eq:pseudo-relativistic}
\end{equation}
instead of the non-relativistic kinetic energy, but this generalization (and the appropriate modifications on $\cE_{\rm TF}$ and on the assumptions on the potential) will not be discussed further. Finally, we notice that the assumption on the magnetic potential $A$ is probably far from optimal, but it already covers the physical case of a potential in 3D satisfying $\nabla\cdot A=0$ and whose magnetic field $B=\nabla\times A$ is square-integrable (since then $A\in L^6(\R^3)$ by the Sobolev inequality). 

Results of the form of~\eqref{eq:CV_GS_energy} have been proved for particular models. For atoms as in the work of Lieb and Simon~\cite{LieSim-73,LieSim-77}, we have 
\begin{equation}
d=3,\qquad A(x)=0,\qquad V(x)=-\frac{1}{t|x|},\qquad w(x)=\frac{1}{|x|},
\label{eq:atoms}
\end{equation}
where $t=\lim(N/Z)$ is the limiting proportion of electrons and protons. In~\cite{LieYau-87}, Lieb and Yau studied pseudo-relativistic stars but their results also apply to the simpler non-relativistic model which corresponds to 
\begin{equation}
d=3,\qquad A(x)=0,\qquad V(x)=0,\qquad w(x)=-\frac{1}{|x|}. 
 \label{eq:stars}
\end{equation}

In Section~\ref{sec:proof_CV_energy}, we provide an elementary proof of Theorem~\ref{thm:CV_GS_energy} which is very much inspired of~\cite{LieYau-87} and was recently written for bosons in~\cite{Lewin-15}.

\subsubsection*{Coherent states and Wigner functions}
The main results of the paper concern the convergence of states, which requires more subtle tools. We will express it using coherent states and Wigner functions, but other choices are possible. 

Let $f$ be any fixed normalized real-valued function in $L^2(\R^d)$. For every fixed $(x,p)$ in the phase space $\R^d\times\R^d$, we introduce the \emph{coherent state}
\begin{equation}
f_{x,p}^ \hbar(y)=\hbar^{-\frac{d}{4}}f\left(\frac{y-x}{\sqrt\hbar}\right)e^{i \frac{p\cdot y}{\hbar}},
\label{eq:f_x_p}
\end{equation}
where we recall that $\hbar=N^{-1/d}$. Then we have the resolution of the identity
\begin{equation}
\frac{1}{(2\pi\hbar)^{d}}\int_{\R^d}\int_{\R^d}|f^\hbar_{x,p}\rangle\langle f^\hbar_{x,p}|\,dx\,dp=1
\label{eq:resolution_of_identity_intro}
\end{equation}
in $L^2(\R^d)$. A typical choice is $f$ a Gaussian, but any real-valued $f$ can indeed be used. For any such $f$ and a fermionic $N$-particle state $\Psi_N$, we introduce the corresponding $k$-particle Husimi function~\cite{Husimi-40,Takahashi-86,ComRob-12}
\begin{multline}
m^{(k)}_{f,\Psi_N}(x_1,p_1,...,x_k,p_k)\\
:=\pscal{\Psi_N,a^*(f^\hbar_{x_1,p_1})\cdots a^*(f^\hbar_{x_k,p_k})a(f^\hbar_{x_k,p_k})\cdots a(f^\hbar_{x_1,p_1})\Psi_N},
\label{eq:def_m_k_Psi_intro}
\end{multline}
for $k=1,...,N$, where $a$ and $a^*$ are the fermionic annihilation and creation operators. In practice $f$ is a very well localized function and the measure $m^{(k)}_{f,\Psi_N}$ describes how many particles are in the $k$ semi-classical boxes with length scale $\sqrt\hbar$,
centered at $(x_1,p_1)$,...,$(x_k,p_k)$ in the phase space $\R^d\times\R^d$. Two equivalent formulas are
\begin{align}
&m^{(k)}_{f,\Psi_N}(x_1,p_1,...,x_k,p_k)\nn\\
&\quad=\frac{N!}{(N-k)!}\pscal{\Psi_N,\big(P_{x_1,p_1}^\hbar\otimes\cdots \otimes P_{x_k,p_k}^\hbar\otimes \1_{N-k}\big)\,\Psi_N}_{L^2(\R^{dN})}
\label{eq:def_m_k_Psi_tensor_product}\\
&\quad=\frac{N!}{(N-k)!}\int_{\R^{d(N-k)}}\left|\pscal{f^\hbar_{x_1,p_1}\otimes\cdots\otimes f^\hbar_{x_k,p_k},\Psi_N(\cdot,\by)}_{L^2(\R^{dk})}\right|^2\,d\by\label{eq:def_m_k_Psi_integral}
\end{align}
where $P_{x,p}^\hbar:=|f^\hbar_{x,p}\rangle\langle f_{x,p}^\hbar|$ is the orthogonal projection onto $f^\hbar_{x,p}$. As will be proved below in Lemma~\ref{lem:prop_m}, we have 
$$0\leq m^{(k)}_{f,\Psi_N}\leq 1$$
and
$$\frac{1}{(2\pi)^{dk}}\int_{\R^{2dk}}m^{(k)}_{f,\Psi_N}=\frac{N(N-1)\cdots (N-k+1)}{N^k}\underset{N\to\ii}{\longrightarrow}1,$$
for all $k\geq1$.

Next we turn to the $k$-particle Wigner function which is defined as in~\cite{LioPau-93,FroGraSch-07} by 
\begin{multline}
\mathscr{W}_{\Psi_N}^{(k)}(x_1,p_1,...,x_k,p_k):=\int_{\R^{dk}}\int_{\R^{d(N-k)}}e^{-i\sum_{\ell=1}^kp_\ell\cdot y_\ell}\times\\ \times\Psi_N(x_1+ \hbar y_1/2,\cdots, x_k+ \hbar y_k/2,x_{k+1},...,x_N)\times\\ \times\overline{\Psi_N(x_1-\hbar y_1/2,\cdots, x_k-\hbar y_k/2,x_{k+1},...,x_N)}\,dy_1\cdots dy_k\,dx_{k+1}\cdots dx_N.
 \label{eq:def_Wigner_k}
\end{multline}
Contrary to $m_{f,\Psi_N}^{(k)}$, the Wigner function $\mathscr{W}_{\Psi_N}^{(k)}$ is not necessarily positive, but it will have the same positive limit as $m_{f,\Psi_N}^{(k)}$ in the semi-classical regime.

\subsubsection*{Convergence of states: confined case}
In our main results about the convergence of states, we for simplicity distinguish the confined and unconfined situations, which we state in two separate theorems. We start with the former.

\begin{theorem}[Convergence of states, confined case]\label{thm:CV_states_confined}
Assume that $w$ is even, that $w,|A|^2,V_-\in L^{1+d/2}(\R^d)+L^\ii_\epsilon(\R^d)$ and that
$$V_+\in L^1_{\rm loc}(\R^d),\qquad \lim_{|x|\to\ii}V_+(x)=+\ii.$$
Let $\{\Psi_N\}\subset\bigwedge ^N L^2(\R^d)$ be any sequence such that $\|\Psi_N\|=1$ and
\begin{equation}
\pscal{\Psi_N,H_N\Psi_N}=E(N)+o(N).
\label{eq:energy_leading_order}
\end{equation}
Then there exists a subsequence $\{N_j\}$ and a probability measure $\mathscr P$ on the set 
$$\cM=\left\{0\leq\rho\in L^1(\R^d)\cap L^{1+2/d}(\R^d)\ :\ \int_{\R^d}\rho=1,\ \cE^{V}_{\rm TF}(\rho)=e_{\rm TF}^V(1)\right\}$$ 
of all the minimizers of the Thomas-Fermi functional, such that the following limit holds:
\begin{equation}
m^{(k)}_{f,\Psi_{N_j}}(x_1,p_1,...,x_k,p_k)\to \int_\cM \prod_{\ell=1}^k\underbrace{\1\left(|p_\ell+A(x_\ell)|^2\leq c_{\rm TF}\,\rho(x_\ell)^{2/d}\right)}_{=m_\rho(x_\ell,p_\ell)}\,d{\mathscr P}(\rho)
\label{eq:CV_m_k}
\end{equation}
weakly in $L^1(\R^{2dk})$ and weakly-$\ast$ in $L^\ii(\R^{2dk})$, for all $k\geq1$ and every real-valued normalized $f\in L^2(\R^d)$ or, in other words,
\begin{equation}
\int_{\R^{2dk}}m^{(k)}_{f,\Psi_{N_j}}\phi \to \int_\cM \left(\int_{\R^{2dk}} (m_\rho)^{\otimes k}\phi\right)\,d{\mathscr P}(\rho)
\label{eq:CV_m_k_explained}
\end{equation}
for every test function $\phi\in L^1(\R^{2dk})+L^\ii(\R^{2dk})$. For the Wigner function $\mathscr{W}_{\Psi_{N_j}}^{(k)}$ defined in~\eqref{eq:def_Wigner_k}, we have the same limit
\begin{equation}
\int_{\R^{2dk}}\cW^{(k)}_{\Psi_{N_j}}\phi \to \int_\cM \left(\int_{\R^{2dk}} (m_\rho)^{\otimes k}\phi\right)\,d{\mathscr P}(\rho),
\label{eq:CV_W_k_explained}
\end{equation}
this time for every test function $\phi$ such that 
\begin{equation}
\partial_{x_1}^{\alpha_1}\cdots \partial_{x_k}^{\alpha_k}\partial_{p_1}^{\beta_1}\cdots \partial_{p_k}^{\beta_k}\phi\in L^\ii(\R^{2dk}),\qquad\max(\alpha_j,\beta_j)\leq 1.
\label{eq:ass_Hwang}
\end{equation}
Furthermore, we have the convergence of the $k$-particle probability density 
\begin{equation}
\int_{\R^d}\cdots\int_{\R^d}|\Psi_{N_j}(x_1,...,x_{N_j})|^2\,dx_{k+1}\cdots dx_{N_j} \to \int_\cM \prod_{j=1}^k\rho(x_j)\;d{\mathscr P}(\rho)
\label{eq:CV_rho_k}
\end{equation}
weakly in $L^1(\R^{d})\cap L^{1+\frac{2}{d}}(\R^{d})$ for $k=1$, and weakly-$\ast$ in the sense of measures for $k\geq2$. Finally, we have the convergence of the $k$-particle kinetic energy density
\begin{multline}
\int_{\R^d}\cdots\int_{\R^d}\Big|\cF_{\hbar}[\Psi_{N_j}](p_1,...,p_{N_j})\Big|^2\,dp_{k+1}\cdots dp_{N_j}\\
\to \int_\cM \prod_{\ell=1}^k\left|\left\{\rho\geq |p_\ell+A|^d c_{\rm TF}^{-d/2}\right\}\right|\;d{\mathscr P}(\rho),
\label{eq:CV_t_k}
\end{multline}
weakly-$\ast$ in the sense of measures for $k\geq1$.
\end{theorem}

In the statement,
\begin{equation}
\cF_\hbar[f](p):=\frac{1}{(2\pi\hbar)^{d/2}}\int_{\R^d}f(x) e^{-i\frac{p\cdot x}{\hbar}}\,dx
\label{eq:Fourier_hbar}
\end{equation}
is the $\hbar$-dependent Fourier transform. Later we also use the notation $\widehat{f}:=\cF_1[f]$ for the unscaled Fourier transform.

The condition $\max(\alpha_j,\beta_j)\leq 1$ in~\eqref{eq:ass_Hwang} means that all the components of the $d$-dimensional multi-indices $\alpha_j$ and  $\beta_j$ are $\leq1$. The condition on $\phi$ is taken from~\cite{Hwang-87} but it can be replaced by any other for which the Calderon-Vaillancourt theorem holds true.

The result says that, in the limit $N\to\ii$, the many-body approximate minimizers $\Psi_N$ become purely semi-classical to leading order and that the corresponding semi-classical measures are a convex combination of factorized states involving the Vlasov minimizers $m_\rho$ with $\rho\in\cM$. Note that if the Thomas-Fermi energy has a unique minimizer $\rho_0$, then there is no need to extract subsequences and the probability measure $\mathscr{P}$ has to be a delta measure at $\rho_0$. In particular, all the limits are now factorized, which corresponds to independent probabilities. The convergence of the $k$-particle densities~\eqref{eq:CV_rho_k} and kinetic energy densities~\eqref{eq:CV_t_k} follows easily from the convergence~\eqref{eq:CV_m_k} of the Husimi measures and the fact that the system is confined (and from the Lieb-Thirring inequality~\cite{LieThi-75,LieThi-76,LieSei-09} for the one-particle density).

The simplest fermionic trial states $\Psi_N$ are the Slater determinants $\Psi_N=(N!)^{-1/2}\det(\psi_i(x_j))$ (a.k.a.~Hartree-Fock states) where $\psi_1,...,\psi_N$ form an orthonormal system~\cite{LieSim-77,Lions-87,Solovej-03}. The one-particle Husimi measure of a Slater determinant is given by
$$m^{(1)}_{f,\Psi_N}(x,p)=\sum_{i=1}^N\big|\langle\psi_i,f^\hbar_{x,p}\rangle\big|^2.$$
It can be proved that $\Psi_N$ satisfies~\eqref{eq:energy_leading_order} if $m^{(1)}_{f,\Psi_N}(x,p)$ is converging to a Vlasov minimizer $m_\rho$ with $\rho\in\cM$ (see Lemma~\ref{lem:semi-classics} below). For a Slater determinant, the $k$-particle semi-classical measures are always factorized in the semi-classical limit and we then end up with a Dirac delta $\mathscr{P}=\delta_\rho$. Our proof of the convergence of the energy (Theorem~\ref{thm:CV_GS_energy}) will indeed rely on the Hartree-Fock problem.

An equivalent way to formulate the limits in Theorem~\ref{thm:CV_states_confined} is in terms of quantization of observables. In other words, for every test function $\phi$ we associate an operator ${\rm Op}^{N,\hbar}(\phi)$ acting on $\bigwedge^N L^2(\R^d)$ and we look at the limit of $\pscal{\Psi_N,{\rm Op}^{N,\hbar}(\phi)\Psi_N}$. The quantization associated with the coherent states $f_{x,p}^\hbar$ is defined by 
\begin{multline*}
{\rm Op}^{N,\hbar}_{f}(\phi):=\int_{(\R^d\times\R^d)^k}\phi(x_1,...,p_k)a^*(f_{x_1,p_1}^\hbar)\cdots a^*(f^\hbar_{x_k,p_k})\times\\
 \times a(f^\hbar_{x_k,p_k})\cdots a(f^\hbar_{x_1,p_1})\,dx_1\cdots dp_k.
\end{multline*}
Similarly, associated with the Wigner function there is an operator defined in terms of its integral kernel by
\begin{multline*}
{\rm Op}^{N,\hbar}_{\rm Weyl}(\phi)(x_1,...,x_N,y_1,...,y_N)\\:={N\choose k}^{-1}\sum_{1\leq i_1<\cdots<i_k\leq N}{\rm Op}^{k,\hbar}_{\rm Weyl}(\phi)(x_{i_1},...,x_{i_k},y_{i_1},...,y_{i_k})
\end{multline*}
where
\begin{multline*}
{\rm Op}^{k,\hbar}_{\rm Weyl}(\phi)(x_1,...,x_k,y_1,...,y_k)\\:=
\hbar^{-dk}\int_{\R^{dk}}\phi\left(\frac{x_{1}+y_{1}}2,p_1,...,\frac{x_{k}+y_{k}}2,p_k\right)e^{\frac{i}{\hbar}\sum_{j=1}^kp_j\cdot(x_{j}-y_{j})}\,dp_1\cdots dp_k.
\end{multline*}
Then the theorem gives the limit 
$$\pscal{\Psi_{N_j},{\rm Op}^{N_j,N_j^{-1/d}}_{f/\rm Weyl}(\phi)\Psi_{N_j}}\underset{j\to\ii}{\longrightarrow}\int_\cM \left(\int_{\R^{2dk}} m_\rho^{\otimes k}\phi\right)\,d{\mathscr P}(\rho)$$ 
for any fixed $\phi$ in the function spaces mentioned in the statement.

\subsubsection*{Convergence of states: unconfined case}
In the unconfined case we have a similar result, except that the limits are \emph{a priori} local. Since some of the particles can escape to infinity, our result will involve the minimizers of the problems $e^V_{\rm TF}(\lambda)$ for a mass $0\leq \lambda\leq1$.

\begin{theorem}[Convergence of states, unconfined case]\label{thm:CV_states_unconfined}
Assume that $w$ is even and that 
$$w,|A|^2,V\in  L^{1+d/2}(\R^d)+L^\ii_\epsilon(\R^d).$$ 
Let $\{\Psi_N\}\subset\bigwedge ^N L^2(\R^d)$ be any sequence such that $\|\Psi_N\|=1$ and
$$\pscal{\Psi_N,H_N\Psi_N}=E(N)+o(N).$$
Then there exists a subsequence $\{N_j\}$ and a probability measure $\mathscr P$ on the set 
\begin{multline*}
\cM=\bigg\{0\leq\rho\in L^1(\R^d)\cap L^{1+2/d}(\R^d)\ :\ \int_{\R^d}\rho\leq 1,\\
\cE^{V}_{\rm TF}(\rho)=e_{\rm TF}^V\Big(\int_{\R^d}\rho\Big)=e_{\rm TF}^V(1)-e_{\rm TF}^0\Big(1-\int_{\R^d}\rho\Big)\bigg\} 
\end{multline*}
which coincides with the set of all the possible weak limits of minimizing sequences for the Thomas-Fermi problem, such that 
\begin{itemize}
 \item the limit~\eqref{eq:CV_m_k_explained} for $m^{(k)}_{f,\Psi_{N_j}}$ holds for every $\phi\in L^1(\R^{2dk})+L^\ii_\eps(\R^{2dk})$;
 \item the limit~\eqref{eq:CV_W_k_explained} for $\cW^{(k)}_{\Psi_{N_j}}$ holds for every $\phi$ satisfying~\eqref{eq:ass_Hwang} and tending to zero at infinity;
 \item the limit~\eqref{eq:CV_rho_k} for $\rho^{(k)}_{\Psi_{N_j}}$ holds weakly in $L^1_{\rm loc}(\R^{d})\cap L^{1+\frac{2}{d}}(\R^{d})$ for $k=1$, and weakly-$\ast$ on $C^0_0(\R^{dk})$ instead of $C^0(\R^{dk})$ for $k\geq2$.
\end{itemize}
If $\int_{\R^d}\rho=1$ for $\mathscr{P}$-almost-every $\rho\in\cM$, then these limits hold as in Theorem~\ref{thm:CV_states_confined}.
\end{theorem}

In the unconfined case some particles may be lost at infinity (if not all), and the limiting minimizing densities $\rho$ might not be probability measures. Nevertheless, the result says that the remaining particles must solve the minimization problem $e^V_{\rm TF}(\int\rho)$, corresponding to the fraction $\int_{\R^d}\rho$ of the $N$ particles which have not escaped to infinity. Furthermore, if no particle is lost ($\int_{\R^d}\rho=1$ on $\cM$), then the convergence is the same as in the confined case.

Note that the weak limit of the $k$-particle kinetic energy density~\eqref{eq:CV_t_k} is unknown in the unconfined case, due to the lack of control in the $x$-variable.
However, it follows from~\eqref{eq:CV_m_k_explained} that the amount of kinetic energy in the limit cannot exceed the Vlasov one
$$
\liminf_{j\to\ii}\frac{\pscal{\Psi_{N_j},(-\Delta)\Psi_{N_j}}}{N_j^{1+\frac2d}}\geq \int_\cB\left(\frac{d }{d+2}c_{\rm TF}\int_{\R^d}\rho^{1+\frac{2}{d}}\right)\;d{\mathscr P}(\rho).
$$

\subsubsection*{Discussion and strategy of proof}
Our theorems~\ref{thm:CV_states_confined} and~\ref{thm:CV_states_unconfined} seem to be the first dealing with the convergence to the Thomas-Fermi problem for general potentials $A,V,w$. Given that the Thomas-Fermi model is at the core of many approximate models in chemistry and material science~\cite{EllLeeCanBur-08}, this result is important for applications. For atoms and stars, our result provides the following information.

\begin{example}[Large atoms]
For the choice~\eqref{eq:atoms}, the Thomas-Fermi energy is
$$\int_{\R^3}\left(\frac35 c_{\rm TF} \rho(x)^{5/3}-\frac{\rho(x)}{t|x|}\right)\,dx+\frac12 \int_{\R^3}\int_{\R^3}\frac{\rho(x)\,\rho(y)}{|x-y|}\,dx\,dy$$
and the set $\cM$ consists of a unique $\rho_t$ which has the mass $\int_{\R^3}\rho_t=\min(1,1/t)$. Theorem~\ref{thm:CV_states_unconfined} therefore provides the strong convergence to $m_{\rho_t}$ for $0<t\leq 1$ and the weak convergence to $m_{\rho_1}$ for $t>1$. This was proved in~\cite{LieSim-73,LieSim-77b,Lieb-81b}.
\end{example}

\begin{example}[Non-relativistic stars]
For stars~\eqref{eq:stars}, the corresponding Thomas-Fermi energy is
$$\frac35 c_{\rm TF}\int_{\R^3} \rho(x)^{5/3}\,dx-\frac12 \int_{\R^3}\int_{\R^3}\frac{\rho(x)\,\rho(y)}{|x-y|}\,dx\,dy$$
and we have $\cM=\{0\}\cup\big\{\rho_0(\cdot-\tau),\ \tau\in\R^3\big\}$ for some unique $\rho_0$ with $\int_{\R^3}\rho_0=1$~\cite{AucBea-71,AucBea-71b,Friedman,Lions-81b,Lions-84,LieYau-87}. Therefore, if we can find a translation of the system for which the limiting semi-classical measures are non trivial, the limit can only involve the Thomas-Fermi minimizing density $\rho_0$. The convergence of states seems to be new. It was only discussed for a locally perturbed model in~\cite[Thm.~6]{LieYau-87}. As mentioned previously, we have a similar convergence result for pseudo-relativistic stars.
\end{example}

Our main new tool in this paper is a fermionic (weak) version of the de Finetti-Hewitt-Savage theorem for \emph{classical} measures (Theorem~\ref{thm:deFinetti_weak} below), which implies that the weak limits of the semi-classical measures must always be a combination of factorized probabilities (Theorem~\ref{thm:subsequences_m_k} below). We follow here ideas of~\cite{LewNamRou-14}, where a weak version of the \emph{quantum} de Finetti theorem was introduced for the mean-field limit of bosons. The fact that we do not need a quantum version for fermions is of course due to the semi-classical feature of our model. The Pauli principle only persists in the constraint that $0\leq m\leq1$. We refer to~\cite{Rougerie-15} for a general presentation of classical and quantum de Finetti theorems, with applications to the mean-field limit for bosons, and to~\cite{Golse-13,AmmNie-08} for the time-dependent problem.

\subsubsection*{Organization of the paper}
The rest of the paper is devoted to the proofs of our results. In Section~\ref{sec:measures} we give the main properties of the Husimi and Wigner measures. The main result there is Theorem~\ref{thm:subsequences_m_k} in which we explain that a sequence $\{\Psi_N\}$ has always convergent semi-classical measures (up to extraction of a subsequence), with a limit that is a convex combination of factorized `fermionic'' measures on the phase space:
$$m^{(k)}_{f,\Psi_{N_j}}\wto \int_{\substack{0\leq\mu\leq 1\\ \int_{\R^{2d}}\mu\leq (2\pi)^d}} \mu^{\otimes k}\;d{\rm P}(\mu),\qquad \forall k\geq1.$$
Sections~\ref{sec:proof_CV_energy} and~\ref{sec:proof_CV_states} are then devoted to the proof of Theorems~\ref{thm:CV_GS_energy},~\ref{thm:CV_states_confined} and~\ref{thm:CV_states_unconfined}. The proof of Theorem~\ref{thm:CV_GS_energy} is based on ideas of Lieb and Yau~\cite{LieYau-87}, while that of Theorem~\ref{thm:CV_states_confined} follows easily from our fermionic de Finetti-Hewitt-Savage theorem. The proof of Theorem~\ref{thm:CV_states_unconfined} is more tedious and relies on the techniques introduced in~\cite{Lewin-11,LewNamRou-14}.

\subsubsection*{Acknowledgement} 
M.L. and J.P.S acknowledge financial support from the European Research Council (Grant Agreements MNIQS 258023 and MASTRUMAT 321029). S.F. acknowledges support from a Danish research council Sapere
Aude grant. This work was started when the authors were at the Centre \'Emile Borel of the Institut Henri Poincaré in Paris in 2013. Part of this work was done when S.F. was a visiting professor at the University Paris-Dauphine.

\section{Fermionic semi-classical measures}\label{sec:measures}

This section is devoted to the study of the general properties of the Husimi measures $m^{(k)}_{f,\Psi_N}$ which we have defined in~\eqref{eq:def_m_k_Psi}, and to the relation with the Wigner measures. 
The main result of the section is a general theorem about the properties of the measures obtained in the limit $N\to\ii$, whose proof will be given later in Section~\ref{sec:proof_thm_subsequences_m_k}. 

\begin{theorem}[Convergence to factorized fermionic measures on phase space]\label{thm:subsequences_m_k}
Let $\Psi_N$ be a sequence of normalized fermionic functions, $\hbar=N^{-1/d}$ and $m^{(k)}_{f,\Psi_N}, \mathscr{W}_{\Psi_N}^{(k)}$ be defined as in~\eqref{eq:def_m_k_Psi_intro} and~\eqref{eq:def_Wigner_k}. 
Then, there exists a subsequence $N_j$ and a probability measure ${\rm P}$ on the set 
$$\cB=\left\{\mu\in L^1(\R^{2d})\ :\ 0\leq \mu\leq 1,\ (2\pi)^{-d}\int_{\R^{2d}} \mu\leq 1\right\}$$
such that, for every $k\geq1$, 
\begin{equation}
\int_{{\R^{2dk}}}m^{(k)}_{f,\Psi_{N_j}}\phi\to\int_\cB \left(\int_{{\R^{2dk}}}\mu^{\otimes k}\phi\right)\;d{\rm P}(\mu),
\label{eq:def_weak_limit}
\end{equation}
for every normalized real-valued function $f\in L^2(\R^d)$ and every $\phi\in L^1({\R^{2dk}})+L^\ii_\epsilon({\R^{2dk}})$, and
\begin{equation}
 \int_{\R^{2dk}}\mathscr{W}_{\Psi_{N_j}}^{(k)}\phi\to \int_\cM \left(\int_{\R^{2dk}} \mu^{\otimes k}\phi\right)\,d{\rm P}(\mu)
\label{eq:def_weak_limit_Wigner2}
\end{equation}
for every function $\phi$ tending to zero at infinity, satisfying~\eqref{eq:ass_Hwang}. 

Furthermore, if $\Psi_N$ satifies the kinetic energy bound
\begin{equation}
\pscal{\Psi_N,\left(\sum_{j=1}^N-\Delta_j\right)\Psi_N}\leq CN^{1+2/d},
\label{eq:semi-classical_bound}
\end{equation}
then 
$$\int_{\cB}\left(\int_{\R^{2d}}|p|^2\mu(x,p)\,dx\,dp\right)\,d{\rm P}(\mu)\leq C(2\pi)^{d}$$
(hence $\int_{\R^{2d}} |p|^2\mu(x,p)\,dx\,dp$ is finite $\rm P$-almost surely) and the $k$-particle densities converge weakly
\begin{equation}
\int_{\R^{dN}}U(x_1,...,x_k)|\Psi_n(x_1,...,x_N)|^2dx_1\cdots dx_N\to\int_\cB \left(\int_{\R^{dk}}\rho_\mu^{\otimes k} U\right)\;d{\rm P}(\mu)
\label{eq:def_weak_limit_density}
\end{equation}
for every $k\geq1$ and every $U\in L^{1+d/2}(\R^{d})+L^\ii_\epsilon(\R^{d})$ for $k=1$ and $U\in C^0_0(\R^{dk})$ for $k\geq2$, where
$$\rho_\mu(x)=\frac{1}{(2\pi)^d}\int_{\R^d}\mu(x,p)\,dp.$$
\end{theorem}

The result says that, whatever the sequence $\{\Psi_N\}$, the semi-classical measures always get factorized in the limit $N\to\ii$ or, more precisely, they are a convex combination of factorized phase-space measures uniformly bounded by $1$ and with a mass $\leq1$. In the rest of this section, we first derive some elementary properties of the measures $m^{(k)}_{f,\Psi_{N}}$, then we state a general convergence theorem in the spirit of the classical de Finetti-Hewitt-Savage theorem, before we provide the detailed proof of Theorem~\ref{thm:subsequences_m_k} in Section~\ref{sec:proof_thm_subsequences_m_k}.

\subsection{Measures on phase space for $N$-body states}

In this section, we recall the elementary properties of the measures $m^{(k)}_{f,\Psi_{N}}$. Some of the results of this section are well-known, but we gather them all here for the convenience of the reader.

\subsubsection*{Coherent states}
As before, we fix a real-valued normalized function $f\in L^2(\R^d)$ and define 
\begin{equation}
f_{x,p}^ \hbar(y)=\hbar^{-\frac{d}{4}}f\left(\frac{y-x}{\sqrt\hbar}\right)e^{i \frac{p\cdot y}{\hbar}},
\label{eq:f_x_p2}
\end{equation}
as well as
\begin{equation}
f^\hbar:=f^\hbar_{0,0}=\hbar^{-d/4}f(\cdot/\sqrt{\hbar})\quad\text{and}\quad g^\hbar:=\cF_\hbar[f^\hbar]=\hbar^{-d/4}\widehat{f}(\cdot/\sqrt{\hbar}).
\label{eq:def_f_hbar_g_hbar}
\end{equation}
We will later have to take $\hbar=N^{-1/d}$ but for the moment $\hbar$ could be arbitrary. We recall the resolution of the identity
\begin{equation}
\frac{1}{(2\pi\hbar)^{d}}\int_{\R^d}\int_{\R^d}|f^\hbar_{x,p}\rangle\langle f^\hbar_{x,p}|\,dx\,dp=1.
\label{eq:resolution_of_identity}
\end{equation}
The latter follows from the useful formulas
\begin{align}
\langle f^\hbar_{x,p},u\rangle&=\int_{\R^d}f^\hbar(y-x)u(y)e^{-i \frac{p\cdot y}{\hbar}}\,dy\nn\\
&=(2\pi\hbar)^{d/2}\cF_\hbar\left[f^\hbar(\cdot-x)u\right](p)\label{eq:scalar_product_coherent_state}\\
&=\int_{\R^d}g^\hbar(k-p)\cF_\hbar[u](k)e^{-i \frac{k\cdot x}{\hbar}}\,dk\nn\\
&=(2\pi\hbar)^{d/2}\cF_\hbar\left[g^\hbar(\cdot-p)\cF_\hbar[u]\right](x)\nn
\end{align}
which imply immediately that 
\begin{align*}
\frac{1}{(2\pi\hbar)^{d}}\int_{\R^d}\int_{\R^d}|\langle f^\hbar_{x,p},u\rangle|^2dx\,dp
&=\int_{\R^d}\left(\int_{\R^d}\left|\cF_\hbar\left[f^\hbar(\cdot-x)\,u\right](p)\right|^2\,dp\right)dx\nn\\
&=\int_{\R^d}\left(\int_{\R^d}\left|f^\hbar(y-x)u(y)\right|^2\,dy\right)dx\nn\\
&=\|f^\hbar\|_{L^2(\R^d)}^2\norm{u}_{L^2(\R^d)}^2=\norm{u}_{L^2(\R^d)}^2.
\end{align*}

\subsubsection*{Semi-classical measures on phase space}
Next, we recall the definition~\eqref{eq:def_m_k_Psi_intro} of the $k$-particle semi-classical (Husimi) measure
\begin{multline}
m^{(k)}_{f,\Psi_N}(x_1,p_1,...,x_k,p_k)\\
:=\pscal{\Psi_N,a^*(f^\hbar_{x_1,p_1})\cdots a^*(f^\hbar_{x_k,p_k})a(f^\hbar_{x_k,p_k})\cdots a(f^\hbar_{x_1,p_1})\Psi_N},
\label{eq:def_m_k_Psi}
\end{multline}
for $k=1,...,N$. Here $a^*(f)$ and $a(f)$ are, respectively, the fermionic creation and annihilation operators, satisfying the Canonical Anticommutation Relations~\cite{DerGer-13}
\begin{equation}
\begin{cases}
  a^*(f)a(g)+a(g)a^*(f)=\pscal{g,f},\\
  a^*(f)a^*(g)+a^*(g)a^*(f)=0.
\end{cases}
\label{eq:CAR}
\end{equation}
See~\eqref{eq:def_m_k_Psi_tensor_product} and~\eqref{eq:def_m_k_Psi_integral} for two other formulas for $m^{(k)}_{f,\Psi_N}$.
The following is a simple consequence of these formulas and of the fact that $\Psi_N$ is an antisymmetric wave function.
  
\begin{lemma}[Elementary properties of the phase space measures]\label{lem:prop_m}
Let $f\in L^2(\R^d,\R)$ and $\Psi_N\in\bigwedge_1^NL^2(\R^d)$ be two normalized functions. Then for every $1\leq k\leq N$, the function $m^{(k)}_{f,\Psi_N}$ defined in~\eqref{eq:def_m_k_Psi} is symmetric and satisfies
\begin{equation}
0\leq m^{(k)}_{f,\Psi_N}\leq 1\quad\text{a.e. on $\R^{2dk}$},
\label{eq:m_k_Psi_unif_bound}
\end{equation}
\begin{multline}
\frac{1}{(2\pi)^{dk}}\int_{\R^{2dk}}m^{(k)}_{f,\Psi_N}(x_1,p_1,...,x_k,p_k)\,dx_1\cdots dp_k\\=N(N-1)\cdots (N-k+1)\hbar^{dk},
\label{eq:m_k_Psi_normalization}
\end{multline}
and
\begin{multline}
\frac{1}{(2\pi)^{d}}\int_{\R^{2d}}m^{(k)}_{f,\Psi_N}(x_1,p_1,...,x_k,p_k)\,dx_k\,dp_k\\=\hbar^{d}(N-k+1) m^{(k-1)}_{f,\Psi_N}(x_1,p_1,...,x_{k-1},p_{k-1}).
\label{eq:m_k_Psi_marginal}
\end{multline}
\end{lemma}

\begin{proof}
The symmetry of $m^{(k)}_{f,\Psi_N}$ follows from the fact that two creation (resp. annihilation) operators anti-commute due to~\eqref{eq:CAR} or, equivalently, from the formula~\eqref{eq:def_m_k_Psi_tensor_product} and the fact that $\Psi_N$ is anti-symmetric. Similarly, the uniform bound~\eqref{eq:m_k_Psi_unif_bound} is a consequence of the operator inequality $0\leq a^*(f)a(f)\leq 1$. The formulas~\eqref{eq:m_k_Psi_marginal} and~\eqref{eq:m_k_Psi_normalization} follow from the representation~\eqref{eq:def_m_k_Psi_tensor_product} and the resolution of the identity~\eqref{eq:resolution_of_identity}, which can be rewritten as $(2\pi\hbar)^{-d}\iint_{\R^{2d}}P^\hbar_{x,p}\,dx\,dp=1$.
\end{proof}

From~\eqref{eq:m_k_Psi_normalization} and~\eqref{eq:m_k_Psi_marginal}, the choice
$\hbar=N^{-1/d}$ 
arises naturally, as it makes $(2\pi)^{-dk}m^{(k)}_{f,\Psi_N}$ a probability measure in the limit $N\to\ii$ for all $k\geq1$.

\begin{remark}
The measures $m^{(k)}_{f,\Psi_N}$ can be defined in a similar manner for bosons. Only the uniform bound~\eqref{eq:m_k_Psi_unif_bound} (which is the expression of Pauli's principle here) will be lost.
\end{remark}

\subsubsection*{Link with $k$-particle densities}

For any fixed (normalized) fermionic function $\Psi_N$, we denote by
\begin{equation}
\rho^{(k)}_{\Psi_N}(x_1,...,x_k):={N\choose k}\int_{\R^d}\cdots\int_{\R^d}|\Psi_N(x_1,...,x_N)|^2\,dx_{k+1}\cdots dx_N
\label{eq:def_rho_k}
\end{equation}
and
\begin{equation}
t^{(k)}_{\Psi_N}(p_1,...,p_k):={N\choose k}\int_{\R^d}\cdots\int_{\R^d}|\cF_\hbar[\Psi_N](p_1,...,p_N)|^2\,dp_{k+1}\cdots dp_N
\label{eq:def_t_k}
\end{equation}
the position and momentum densities for $k$ particles, with $1\leq k\leq N$. For later purposes, we also define the $k$-particle density matrix of $\Psi_N$ which is the operator $\gamma^{(k)}_{\Psi_N}$ with integral kernel 
\begin{multline}
\gamma^{(k)}_{\Psi_N}(x_1,...,x_k;x_1',...,x_k'):={N\choose k}\int_{\R^d}\cdots\int_{\R^d}\Psi_N(x_1,...,x_N)\times\\ \times\overline{\Psi_N(x'_1,...,x'_k,x_{k+1},...,x_N)}\,dx_{k+1}\cdots dx_N.
\label{eq:def_gamma_k}
\end{multline}
All the physical observables can be expressed in terms of these objects only. The following gives a link with the semi-classical measures we have defined in the previous section.

\begin{lemma}[Particle densities and fermionic semi-classical measures]\label{lem:link_PDM}
Let $f$ be any normalized function in $L^2(\R^d,\R)$ and $f^\hbar,g^\hbar$ be defined as in~\eqref{eq:def_f_hbar_g_hbar}. Then we have
\begin{equation}
\frac{1}{(2\pi)^{dk}}\int_{(\R^d)^k}m_{f,\Psi_N}^{(k)}(x_1,p_1,...,x_k,p_k)\,dp_1\cdots dp_k=k!\,\hbar^{dk}\;\rho^{(k)}_{\Psi_N}\ast \big(|f^\hbar|^2\big)^{\otimes k}
\label{eq:link_PDM_x}
\end{equation}
and
\begin{equation}
\frac{1}{(2\pi)^{dk}}\int_{(\R^d)^k}m_{f,\Psi_N}^{(k)}(x_1,p_1,...,x_k,p_k)\,dx_1\cdots dx_k=k!\,\hbar^{dk}\;t^{(k)}_{\Psi_N}\ast \big(|g^\hbar|^2\big)^{\otimes k}
\label{eq:link_PDM_p}
\end{equation}
for any (normalized) fermionic function $\Psi_N$.
\end{lemma}

\begin{proof}
As in~\eqref{eq:scalar_product_coherent_state}, we start by noticing that, for every fixed $\by\in (\R^d)^{N-k}$,
\begin{align*}
&\pscal{f^\hbar_{x_1,p_1}\otimes\cdots\otimes f^\hbar_{x_k,p_k},\Psi_N(\cdot,\by)}_{L^2((\R^d)^k)}\\
&\qquad\qquad=(2\pi\hbar)^{dk}\cF_{\hbar}\left[f^\hbar_{x_1,0}\otimes\cdots\otimes f^\hbar_{x_k,0} \Psi_N(\cdot,\by)\right](p_1,...,p_k)\\
&\qquad\qquad=(2\pi\hbar)^{dk}\cF_{\hbar}\left[g^\hbar_{0,p_1}\otimes\cdots\otimes g^\hbar_{0,p_k} \cF_\hbar[\Psi_N](\cdot,\by)\right](x_1,...,x_k).
\end{align*}
Next we compute the integral of $m^{(k)}_{f,\Psi}$ over the $p_j$'s, using~\eqref{eq:def_m_k_Psi_tensor_product}:
\begin{align*}
&\int_{(\R^d)^k}m_{f,\Psi_N}^{(k)}(x_1,p_1,...,x_k,p_k)\,dp_1\cdots dp_k\\
&\ =k!{N\choose k}\int_{(\R^d)^k}dp_1\cdots dp_k\int_{\R^{d(N-k)}}d\by\left|\pscal{f^\hbar_{x_1,p_1}\otimes\cdots\otimes f^\hbar_{x_k,p_k},\Psi_N(\cdot,\by)}\right|^2\\
&\ =(2\pi\hbar)^{dk}k!{N\choose k}\int_{(\R^d)^k}dp_1\cdots dp_k\int_{\R^{d(N-k)}}d\by\\
&\qquad\qquad\times\left|\cF_{\hbar}\left[f^\hbar_{x_1,0}\otimes\cdots\otimes f^\hbar_{x_k,0} \Psi_N(\cdot,\by)\right](p_1,...,p_k)\right|^2\\
&\ =(2\pi\hbar)^{dk}k!{N\choose k}\int_{\R^{dN}}\left|f^\hbar(y_1-x_1)\cdots f^\hbar(y_k-x_k)\Psi_N(y)\right|^2dy_1\cdots dy_N\\
&\ =(2\pi\hbar)^{dk}k!\;\rho^{(k)}_{\Psi_N}\ast \big(|f^\hbar|^2\big)^{\otimes k}(x_1,...,x_k).
\end{align*}
The proof of~\eqref{eq:link_PDM_p} is similar.
\end{proof}

A simple consequence of~\eqref{eq:link_PDM_p} is the well-known formula for the kinetic energy and its generalization with magnetic field (see, e.g.,~\cite[Eq. (5.20)]{Lieb-81b}):

\begin{corollary}[Kinetic energy]\label{cor:kinetic}
If $f\in H^1(\R^3,\R)$ and $|A|^2\in L^{1+d/2}(\R^d)+L^\ii_\eps(\R^d)$, we have
\begin{multline}
\pscal{\Psi_N,\bigg(\sum_{j=1}^N(-i\hbar\nabla_{j}+A(x_j))^2\bigg)\Psi_N}\\
=\frac{1}{(2\pi\hbar)^{d}}\int_{\R^d}\int_{\R^d}|p+A(x)|^2m_{f,\Psi_N}^{(1)}(x,p)\,dx\,dp-N\hbar\int_{\R^d}|\nabla f|^2\\
+2\Re\pscal{\Psi_N,\bigg(\sum_{j=1}^N(A-A\ast|f^\hbar|^2)(x_j)\cdot (-i\hbar\nabla_j)\bigg)\Psi_N}\\
+\pscal{\Psi_N,\bigg(\sum_{j=1}^N(|A|^2-|A|^2\ast|f^\hbar|^2)(x_j)\bigg)\Psi_N}.
\label{eq:formula_kinetic_energy_magnetic_field}
\end{multline}
\end{corollary}

\begin{proof}
By density we may assume that $\Psi$ and $A$ are regular enough, which allows us to make the following calculation. By definition, we have
\begin{align*}
\pscal{\Psi_N,\left(\sum_{j=1}^N-\hbar^2\Delta_j\right)\Psi_N}&=\int_{\R^d}t^{(1)}_{\Psi_N}(p)|p|^2\,dp\\
&=\frac{1}{(2\pi\hbar)^d}\int_{\R^d}\int_{\R^d}m^{(1)}_{f,\Psi_N}(x,p)|p|^2\,dx\,dp\\
&\qquad+\int_{\R^d}\int_{\R^d}t^{(1)}_{\Psi_N}(p)|g^\hbar(q-p)|^2(|p|^2-|q|^2)\,dp\,dq\\
&=\frac{1}{(2\pi\hbar)^d}\int_{\R^d}\int_{\R^d}m^{(1)}_{f,\Psi_N}(x,p)|p|^2\,dx\,dp\\
&\qquad -\int_{\R^d}\int_{\R^d}t^{(1)}_{\Psi_N}(p)|g^\hbar(q-p)|^2|q-p|^2\,dp\,dq\\
&\qquad -2 \left(\int_{\R^d}pt^{(1)}_{\Psi_N}(p)\,dp\right)\cdot\left(\int_{\R^d}p|g^\hbar(p)|^2\,dp\right).
\end{align*}
The last term vanishes since $f$ is real, hence $g^\hbar$ is even, and we obtain 
\begin{multline}
\pscal{\Psi_N,\bigg(\sum_{j=1}^N-\hbar^2\Delta_j\bigg)\Psi_N}\\
=\frac{1}{(2\pi\hbar)^{d}}\int_{\R^d}\int_{\R^d}|p|^2m_{f,\Psi_N}^{(1)}(x,p)\,dx\,dp-N\hbar\int_{\R^d}|\nabla f|^2.
\label{eq:formula_kinetic_energy}
\end{multline}

With magnetic field, we expand the square 
$$|-i\hbar\nabla+A(x)|^2=-\hbar^2\Delta+A(x)\cdot(-i\hbar\nabla)+(-i\hbar\nabla)\cdot A(x)+|A(x)|^2$$
and apply~\eqref{eq:link_PDM_x} for the term with $|A|^2$ and~\eqref{eq:formula_kinetic_energy} for the term with $-\Delta$. The cross term is calculated using~\eqref{eq:scalar_product_coherent_state} as follows:
\begin{align*}
&\frac{1}{(2\pi\hbar)^{d}}\iint_{\R^d\times\R^d} p\cdot A(x) |\langle u,f_{xp}^\hbar\rangle|^2\,dx\,dp\\
&\qquad\qquad=\iint_{\R^d\times\R^d} p\cdot A(x) |\cF_{\hbar}(f^\hbar_{x,0}u)(p)|^2\,dx\,dp\\
&\qquad\qquad=\int_{\R^d} A(x)\cdot \int_{\R^d}p|\cF_{\hbar}(f^\hbar_{x,0}u)(p)|^2\,dp\,dx\\
&\qquad\qquad=\hbar\int_{\R^d} A(x)\cdot \Im\int_{\R^d}f^\hbar_{x,0}(y)\overline{u(y)}\,\nabla(f^\hbar_{x,0}u)(y)\,dy\,dx\\
&\qquad\qquad=\hbar\int_{\R^d} A(x)\cdot \int_{\R^d}|f^\hbar(y-x)|^2\Im(\overline{u(y)}\,\nabla u(y))\,dy\,dx\\
&\qquad\qquad=\hbar\int_{\R^d} (A\ast|f^\hbar|^2)(y)\cdot\Im(\overline{u(y)}\,\nabla u(y))\,dy.
\end{align*}
\end{proof}

\subsection{Structure of the limiting measures: tight case}\label{sec:dF_tight_case}

For an arbitrary sequence $\Psi_N$ with $N\to\ii$, the functions $(m^{(k)}_{f,\Psi_N})_{N\geq k}$ are bounded in $L^1({\R^{2dk}})\cap L^\ii({\R^{2dk}})$, for every fixed $k$, by Lemma~\ref{lem:prop_m}. It is therefore clear that we can find a subsequence such that $m^{(k)}_{f,\Psi_N}\wto m^{(k)}_{f}$ weakly for every $k\geq1$, in the sense that
\begin{equation*}
\int_{{\R^{2dk}}}m^{(k)}_{f,\Psi_N}\phi\to\int_{{\R^{2dk}}}m^{(k)}_{f}\phi
\end{equation*}
for every $\phi\in L^1({\R^{2dk}})+L^\ii_\epsilon({\R^{2dk}})$. In the limit we obtain a hierarchy of symmetric functions  $(m^{(k)}_f)_{k\geq1}$. The purpose of this section and of the following is to study the properties of these limiting hierarchies. In general, the limiting functions $m^{(k)}_f$ are not probability densities because some mass can be lost at infinity. However, if the sequence $(m^{(1)}_{f,\Psi_N})$ is tight, that is,
$$\lim_{R\to\ii}\limsup_{N\to\ii}\int_{|x|+|p|\geq R}m^{(1)}_{f,\Psi_N}(x,p)\,dx\,dp=0,$$
then the $m^{(k)}_{f,\Psi_N}$ are also tight for $k\geq2$ and the limiting $m_f^{(k)}$ are all probability measures. We consider this simpler case in this section. The tightness of the sequence also implies that the compatibility relation~\eqref{eq:m_k_Psi_marginal} is preserved in the limit:
\begin{equation}
\frac1{(2\pi)^{d}}\int_{{\R^{2d}}}m^{(k)}(x_1,p_1,...,x_k,p_k)\,dx_k\,dp_k\\= m^{(k-1)}(x_1,p_1,...,x_{k-1},p_{k-1})
\label{eq:compatibility}
\end{equation}
for all $k\geq1$. The famous de Finetti-Hewitt-Savage theorem deals with the structure of such infinite sequences of symmetric probability measures~\cite{DeFinetti-31,DeFinetti-37,Dynkin-53,HewSav-55,DiaFre-80}. In our situation, the result can be stated as follows.

\begin{theorem}[Fermionic semi-classical measures on phase space]\label{thm:deFinetti_strong}
Let $m^{(k)}$ be a family of symmetric positive densities in $L^1(M^k)$, with $M\subset \R^D$, satisfying 
\begin{equation}
c\int_{M}m^{(k)}(\xi_1,...,\xi_k)\,d\xi_k= m^{(k-1)}(\xi_1,...,\xi_{k-1})
\label{eq:compatibility_bis}
\end{equation}
with $m^{(0)}=1$ and $0\leq m^{(k)}\leq 1$ for all $k\geq1$. Then there exists a Borel probability measure $\rm P$ on the set 
$$\cS:=\left\{\mu\in L^1(M)\ :\ 0\leq \mu\leq 1,\ c\int_M\mu=1\right\}$$ 
such that 
\begin{equation}
m^{(k)}=\int_\cS \mu^{\otimes k}\,d{\rm P}(\mu),
\label{eq:deFinetti}
\end{equation}
for all $k\geq1$.
\end{theorem}

In our case $M=\R^{2d}$ and $c=(2\pi)^{-d}$ but the result is actually true for any Borel set $M$ in $\R^D$ and any $c>0$. The theorem says that infinite exchangeable (i.e.~symmetric) fermionic systems are always convex combination of independent ones.

\begin{proof}
A very similar theorem is proved in~\cite[p.~510--511]{CagLioMarPul-92} (see also ~\cite[Lemma~4]{MesSpo-82}). The usual theorem from~\cite{HewSav-55} furnishes a probability measure $\rm P\in\cP(\cP(M))$ such that~\eqref{eq:deFinetti} holds with $\cS$ replaced by $\cP(M)$. The condition $0\leq \mu\leq 1$ can be characterized by the property that $\mu(A)\leq |A|$ for all $A$ in a countable set of balls, showing that $\cS$ is a measurable set in $\cP(M)$. We therefore only have to prove that this measure $\rm P$ is concentrated on $\cS$. Now, the assumption that $0\leq m^{(k)}\leq1$ implies $m^{(k)}(A^k)\leq |A|^k$ for any Borel set $A\subset M$, and this gives
$$\int_{\cP(M)} \left(\frac{\mu(A)}{|A|}\right)^k\,d{\rm P}(\mu)\leq 1$$
for every $k\geq1$. Taking $k\to\ii$ proves that $\rm P$ is concentrated on the subset of $\cP(M)$ containing all the probability measures $\mu$ such that $\mu(A)\leq |A|$ for all $A$. These measures are absolutely continuous with respect to the Lebesgue measure and the corresponding density is between $0$ and~$1$.
\end{proof}

\subsection{Structure of the limiting measures: general case}

As we have said, in general the limiting functions $m_f^{(k)}$ need not be probability measures, and they need not satisfy the compatibility condition~\eqref{eq:compatibility}. The idea that a de Finetti theorem nevertheless holds true in this case as well seems to have been advertized for the first time in the quantum case in~\cite{LewNamRou-14}. A similar result with a different interpretation had been published before by Ammari and Nier in~\cite{AmmNie-08}. The following is inspired of~\cite[Thm 2.2]{LewNamRou-14}.

\begin{theorem}[Weak fermionic semi-classical measures on phase space]\label{thm:deFinetti_weak}
Let $m_N^{(N)}$ be a sequence of symmetric positive densities in $L^1(M^N)$, with $M\subset \R^D$, and let $m_N^{(k)}$ be its marginals defined recursively as in~\eqref{eq:compatibility_bis}. We assume that $m^{(0)}=1$, that $0\leq m^{(k)}_N\leq 1$ for every $1\leq k\leq N$ and that $m^{(k)}_N\wto m^{(k)}$ weakly in $L^1(M^k)$ and weakly-$\ast$ in $L^\ii(M^k)$ for every fixed $k\geq1$, as $N\to\ii$. Then there exists a Borel probability measure $\rm P$ on the set $\cB:=\{\mu\in L^1(M)\ :\ 0\leq \mu\leq 1,\ c\int_M\mu\leq1\}$ such that 
\begin{equation}
m^{(k)}=\int_\cB \mu^{\otimes k}\,d{\rm P}(\mu),
\label{eq:deFinetti_weak}
\end{equation}
for all $k\geq1$.
\end{theorem}

The theorem says that the natural object obtained in the semi-classical limit of large fermionic systems (with possible lack of compactness) is a Borel probability measure $\rm P$ on the set of all the functions $0\leq \mu\leq1$ on $M$ such that $\int_M\mu\leq1/c$. 

\begin{proof}
The proof follows step by step that of~\cite[Thm 2.2]{LewNamRou-14} in the quantum case and we only outline it. We fix a ball $B_R=\{\xi\in M\ :\ \norm{\xi}\leq R\}$ in the space $M$, with $R$ an arbitrary positive integer. We then look at the probability measures for the particles to be in $B_R$. In particular, we will need to know the exact number of particles in $B_R$, which leads us to introduce the functions
$$
g_{N,R,0}:=c^N\int_{(M\setminus B_R)^{N}}m^{(N)}_N(\xi_1,...,\xi_N)d\xi_{1}\cdots d\xi_N
$$
and
\begin{multline*}
g_{N,R,n}(\xi_1,...,\xi_n)\\:=c^{N-n}{N\choose n} \prod_{j=1}^n\1_{B_R}(\xi_j)\int_{(M\setminus B_R)^{N-n}}m^{(N)}_N(\xi_1,...,\xi_N)d\xi_{n+1}\cdots d\xi_N
\end{multline*}
for $n=1,...,N$, and where $\xi\in M$ is the phase space variable. In words, $g_{N,R,n}$ is the probability density for $n$ particles in $B_R$, under the constraint that the other $N-n$ are all outside of $B_R$. Note that 
\begin{equation}
\sum_{n=0}^Nc^n\int_{(B_R)^n}g_{N,R,n}=c^N\int_{M^N}m^{(N)}_N=1
\label{eq:normalization_localized_state}
\end{equation}
and that the $k$-particle density in $B_R$ can be written as
\begin{align*}
\prod_{j=1}^k\1_{B_R}(\xi_j)m^{(k)}_{N}(\xi_1,...,\xi_k)&=\sum_{n=k}^N\frac{{{N-k}\choose {n-k}}}{{N\choose n}}g_{N,R,n}^{(k)}(\xi_1,...,\xi_k)\\
&=\sum_{n=k}^N\frac{{{n}\choose {k}}}{{N\choose k}}g_{N,R,n}^{(k)}(\xi_1,...,\xi_k)
\end{align*}
where $g_{N,R,n}^{(k)}$ denotes the $k$th marginal of the function $g_{N,R,n}$. 

It is useful to think of the fraction $n/N$ of particles in $B_R$ as an additional variable $0\leq t\leq1$. For fixed $N$ this new parameter takes values in $\Z/N\cap[0,1]$. In the limit $N\to\ii$, the whole interval $[0,1]$ is filled and $t$ will correspond to the radial integration in the ball $\cB$. We therefore introduce the probability measure 
\begin{multline}
dg_{N,R}^{(k)}(t,\xi_1,...,\xi_k):=\sum_{n=k}^N g_{N,R,n}^{(k)}(\xi_1,...,\xi_k)\delta_{n/N}(t)\\+|B_R|^{-k}\left(\sum_{n=0}^{k-1} c^k\int_{(B_R)^k}g_{N,R,n}\right)\delta_0(t)
\label{eq:def_g_N_r}
\end{multline}
on the compact set $[0,1]\times (B_R)^k$. Extracting subsequences, we may assume that $g_{N,R}^{(k)}$ converges weakly to a probability measure $g_{R}^{(k)}$ on $[0,1]\times (B_R)^k$, for every $k\geq1$. Since $R$ is an integer by assumption, the weak convergence can be assumed for every $R$ as well.

Now, from the inequality
$$0\leq \left(\frac{n}{N}\right)^k-\frac{{{n}\choose {k}}}{{N\choose k}}\leq \frac{(k-1)^2}{N-k+1}$$
(see~\cite[Eq.(2.13)]{LewNamRou-14}) and the normalization~\eqref{eq:normalization_localized_state}, we deduce that 
\begin{multline*}
c^k\int_{(B_R)^k}\bigg|\prod_{j=1}^k\1_{B_R}(\xi_j) m^{(k)}_{N}(\xi_1,...,\xi_k)\\ - \int_0^1t^k\,g^{(k)}_{N,R}(t,\xi_1,...,\xi_k)\,dt\bigg|\,d\xi_1\cdots d\xi_k \leq \frac{(k-1)^2}{N-k+1},
\end{multline*}
from which we conclude that
\begin{equation}
\1_{B_R}^{\otimes k} m^{(k)}=\int_0^1 t^k\,dg^{(k)}_{R}(t,\cdot)
\label{eq:link_localized_marginal}
\end{equation}
for every $k,R\geq1$.

Next we remark that
$$t\left(c\int_{B_R}dg_{N,R}^{(k+1)}(\cdot,\xi_{k+1})-dg_{N,R}^{(k)}\right)=t\delta_{k/N}(t)\,g_{N,R}^{(k)}= \frac{k}{N}\delta_{k/N}(t)\,g_{N,R}^{(k)}$$
where the right side converges to $0$ in the sense of measures. Passing to the weak limit $N\to\ii$, we deduce that the limiting hierarchy is consistent for all $0<t\leq1$:
$$ct\int_{B_R}dg_{R}^{(k+1)}(t,\cdot,\xi_{k+1})=t\,dg_{R}^{(k)}(t,\cdot).$$
At $t=0$ this equation tells us nothing. However, we can always change the value of the limit $dg^{(k)}_R$ as we want since the point $t=0$ does not contribute in~\eqref{eq:link_localized_marginal}. 
Similarly to~\eqref{eq:def_g_N_r}, we choose it to be the constant that makes it a probability measure and the consistency at $t=0$ is then obvious.
From the de Finetti-Hewitt-Savage theorem (or, rather, a simple one-parameter version of it which can be proved similarly as sketched in~\cite{LewNamRou-14}), we conclude that there exists a probability measure $\tilde{\rm P}_R$ on $[0,1]\times\cP(B_R)$ such that 
$$dg_{R}^{(k)}(t,\cdot)=\int_{\cP(B_R)}p^{\otimes k}\,d\tilde{\rm P}_R(t,p).$$
Inserting in~\eqref{eq:link_localized_marginal}, we get the representation formula
$$(\1_{B_R})^{\otimes k}m^{(k)}=\int_{[0,1]\times \cP(B_R)}(tp)^{\otimes k}\,d\tilde{\rm P}_R(t,p).$$
Using that $0\leq m^{(k)}\leq1$ we can show similarly as in the proof of Theorem~\ref{thm:deFinetti_strong}, that $\tilde{\rm P}_R$ concentrates on $[0,1]\times\cS$. The right side can then be uniquely interpreted as an integral for a measure ${\rm P}_R$ on $\cB_R:=\{0\leq \mu\leq \1_{B_R}\ :\ \int_{B_R}\mu\leq1/c\}$, with $t$ playing the role of the radial variable:
$$(\1_{B_R})^{\otimes k}m^{(k)}=\int_{\cB_R}\mu^{\otimes k}\,d{\rm P}_R(\mu).$$
Since the previous equality holds for every positive integer $R$, the measure ${\rm P}_R$ is the cylindrical projection of ${\rm P}_{R'}$ for every $R'>R$ and the family $({\rm P}_R)$ defines uniquely a measure ${\rm P}$ over $\cB$. We get the result by passing to the limit $R\to\ii$.
\end{proof}

\subsection{Link with Wigner functions}\label{sec:link_Wigner}
It is well known that the Wigner functions converge to the same limit as the Husimi measures and we quickly recall this here. 
It is indeed a general principle that all regular quantizations are equivalent in the semi-classical limit. 

First, we recall that when $f=(\pi)^{-d/4}e^{-|x|^2/2}$, the Husimi and Wigner measures are related by a convolution:
$$m^{(k)}_{\Psi,f}=\frac{N(N-1)\cdots (N-k+1)}{N^k}\mathscr{W}^{(k)}_{\Psi}\ast \mathscr{G}^\hbar,$$
where $\mathscr{G}^\hbar(x_1,...,p_k)=(\pi\hbar)^{-dk}\exp(-\hbar^{-1}\sum_{j=1}^k|x_j|^2+|p_j|^2)$ is a Gaussian in phase space (see, e.g.~\cite[Prop.~21]{ComRob-12}). It is then clear that the two sequences must have the same weak limits. On the other hand, the Calderon-Vaillancourt theorem from~\cite{Hwang-87} states that
\begin{equation*}
\norm{{\rm Op}_{\rm Weyl}^{k,\hbar} (\phi)}\leq C\sup_{\max(\alpha,\beta)\leq1}\norm{\partial_{x_1}^{\alpha_1}\cdots \partial_{x_k}^{\alpha_k}\partial_{p_1}^{\beta_1}\cdots \partial_{p_k}^{\beta_k} \phi}_{L^\ii(\R^{2dk})},
\end{equation*}
for any function $\phi$ for which the right side is finite. This bound is useful for showing that $\mathscr{W}^{(k)}_{\Psi}$ has weak limits. From this we obtain the estimate
\begin{multline*}
\norm{{\rm Op}_{\rm Weyl}^{N,\hbar} (\phi)-{\rm Op}_{f}^{N,\hbar} (\phi)}\\ \leq C\sup_{\max(\alpha,\beta)\leq1}\norm{\partial_{x_1}^{\alpha_1}\cdots \partial_{x_k}^{\alpha_k}\partial_{p_1}^{\beta_1}\cdots \partial_{p_k}^{\beta_k} (\phi-\phi\ast\mathscr{G}^\hbar)}_{L^\ii(\R^{2dk})}\\
+C\frac{k^2}{N}\sup_{\max(\alpha,\beta)\leq1}\norm{\partial_{x_1}^{\alpha_1}\cdots \partial_{x_k}^{\alpha_k}\partial_{p_1}^{\beta_1}\cdots \partial_{p_k}^{\beta_k} \phi}_{L^\ii(\R^{2dk})}.
\end{multline*}
The term $\phi-\phi\ast\mathscr{G}^\hbar$ can be estimated uniformly by the gradient of $\phi$, using the fundamental theorem of calculus, leading to
\begin{multline*}
\norm{{\rm Op}_{\rm Weyl}^{N,\hbar} (\phi)-{\rm Op}_{f}^{N,\hbar} (\phi)}\\
\leq C_k\sqrt{\hbar}\sup_{\max(\alpha,\beta)\leq2}\norm{\partial_{x_1}^{\alpha_1}\cdots \partial_{x_k}^{\alpha_k}\partial_{p_1}^{\beta_1}\cdots \partial_{p_k}^{\beta_k} \phi}_{L^\ii(\R^{2dk})}.
\end{multline*}
One more derivative of $\phi$ is necessary to prove that the limit of the Wigner function coincides with that of the Husimi measure, but the limit of $\cW_\Psi$ nevertheless holds with only $\max(\alpha,\beta)\leq1$, using the density of smooth test functions.

As a conclusion, these estimates give the convergence of the Wigner functions when tested against any $\phi$ satisfying $\partial_{x_1}^{\alpha_1}\cdots \partial_{x_k}^{\alpha_k}\partial_{p_1}^{\beta_1}\cdots \partial_{p_k}^{\beta_k} \phi \in L^\ii(\R^{2dk})$ for $\max(\alpha,\beta)\leq1$, provided that the convergence of the Husimi measures has been established in the special case $f=(\pi)^{-d/4}e^{-|x|^2/2}$. If the weak convergence for the Husimi measures only holds in $L^1_{\rm loc}(\R^{dk})$, then in addition $\phi$ has to tend to zero at infinity.

\subsection{Proof of Theorem~\ref{thm:subsequences_m_k}}\label{sec:proof_thm_subsequences_m_k}

We are now able to write the proof of Theorem~\ref{thm:subsequences_m_k}.

Let $f\in C^\ii_c(\R^d)$ be a real-valued $L^2$-normalized function. As we have already explained in the beginning of Section~\ref{sec:dF_tight_case}, up to extraction of a subsequence, we have 
\begin{equation}
\int_{M^k}m^{(k)}_{f,\Psi_N}\phi\to\int_{M^k}m^{(k)}_{f}\phi
\label{eq:def_weak_limit3}
\end{equation}
for every $k\geq1$ and every $\phi\in L^1(M^k)+L^\ii_\epsilon(M^k)$. 
By Theorem~\ref{thm:deFinetti_weak}, there exists a probability measure $\rm P$ on $\cB$ such that
$$m^{(k)}_f=\int_\cB \mu^{\otimes k}\,d{\rm P}(\mu)$$
and therefore the convergence~\eqref{eq:def_weak_limit} follows for this $f$. 

Next we establish the independence of the weak convergence with respect to $f$. This is the content of the following

\begin{lemma}[The limiting measures do not depend on $f$]\label{lem:m_do_not_depend_on_f}
Let $\Psi_N$ be a sequence of normalized fermionic functions and let $f\in L^2(\R^d)$ be a normalized function such that $m^{(k)}_{f,\Psi_N}\wto m^{(k)}_f$ weakly as $N\to\ii$ in the sense of~\eqref{eq:def_weak_limit}, for every $k\geq1$. Then $m^{(k)}_{g,\Psi_N}\wto m^{(k)}_f$ weakly, for every normalized $g\in L^2(\R^d)$.
\end{lemma}

\begin{proof}[Proof of Lemma~\ref{lem:m_do_not_depend_on_f}]
It clearly suffices to prove that
$$\lim_{N\to\ii}\int_{M^k}\phi(\xi)\, \big(m^{(k)}_{f,\Psi_N}(\xi)-m^{(k)}_{g,\Psi_N}(\xi)\big)\,d\xi=0$$
for all functions $\phi$ in a dense set. It will be convenient to take $\phi$ of the form
$$\phi(x_1,p_1,...,x_k,p_k)=u_1(x_1)v_1(p_1)\cdots u_k(x_k)v_k(p_k)$$
with $u_j,v_j$ in the Schwartz class. 

First we remark that we can assume that $f$ and $g$ are smooth with compact support. Indeed, using the definition of $m^{(k)}_{f,\Psi_N}$ and that $\|a(f-\tilde{f})\|=\|a^*(f-\tilde{f})\|=\|f-\tilde{f}\|_{L^2}$, we find that
$$\norm{m^{(k)}_{f,\Psi_N}-m^{(k)}_{\tilde{f},\Psi_N}}_{L^\ii(M^k)}\leq 2k \|f-\tilde{f}\|_{L^2}.$$
Therefore we can approximate $f$ and $g$ by smooth functions $\tilde f$ and $\tilde g$ in $L^2(\R^d)$, at the expense of an error which is uniform in $N$. Then it suffices to prove the result for $\tilde f$ and $\tilde g$. For simplicity, we just assume that $f$ and $g$ are smooth for the rest of the proof.

By definition of the phase space measure $m^{(k)}_{f,\Psi_N}$, we have
\begin{multline*}
\frac1{(2\pi)^{dk}}\int_{M^k}\phi(\xi)\, m^{(k)}_{f,\Psi_N}(\xi)\,d\xi\\
=N(N-1)\cdots (N-k+1)\pscal{\Psi_N,{\rm Op}_f^{k,\hbar}(\phi)\otimes\1_{N-k}\,\Psi_N} 
\end{multline*}
where
\begin{equation}
{\rm Op}_f^{k,\hbar}(\phi)=\frac1{(2\pi)^{dk}}\int_{M^k} \phi(x_1,p_1,...,x_k,p_k)\,P_{x_1,p_1}^\hbar\otimes\cdots\otimes P_{x_k,p_k}^\hbar dx_1\cdots dp_k.
\label{eq:formula_phi_tilde}
\end{equation}
We now claim that
\begin{equation}
\norm{{\rm Op}_f^{k,\hbar}(\phi)-{\rm Op}_g^{k,\hbar}(\phi)}\leq C\hbar^{kd}\sqrt\hbar=CN^{-k-\frac1{2d}}
\label{eq:compare_phi_f_g}
\end{equation}
where $C$ depends on the regularity of $f$, $g$, $\phi$ and on $k$, but not on $N$. This will imply
$$\left|\int_{M^k}\phi(\xi)\, \big(m^{(k)}_{f,\Psi_N}(\xi)-m^{(k)}_{g,\Psi_N}(\xi)\big)\,d\xi\right|\leq C\sqrt\hbar=CN^{-\frac1{2d}}$$
and finish the argument. The proof of~\eqref{eq:compare_phi_f_g} is based on well-known techniques from semi-classical analysis. It consists in estimating the difference between the above quantization~\eqref{eq:formula_phi_tilde} of $\phi$ and another quantization which is independent of $f$. We provide the full proof here for the convenience of the reader.

First we remark that for our specific function $\phi$ we can write
$${\rm Op}_f^{k,\hbar}(\phi)=B^\hbar_{f,1}\otimes \cdots\otimes B^\hbar_{f,k}$$
where 
$$B^\hbar_{f,j}=\frac1{(2\pi)^{d}}\int_{M}u_j(x)v_j(p)|f^\hbar_{x,p}\rangle\langle f_{x,p}^\hbar|\,dx\,dp.$$
Since by~\eqref{eq:resolution_of_identity} we have $\|B^\hbar_{f,j}\|\leq \hbar^d\norm{u_j}_{L^\ii(\R^d)}\norm{v_j}_{L^\ii(\R^d)}$ and the same estimate for $B^\hbar_{g,j}$, we deduce that
$$\norm{{\rm Op}_f^{k,\hbar}(\phi)-{\rm Op}_g^{k,\hbar}(\phi)}\leq \hbar^{d(k-1)}\sum_{j=1}^k\|B^\hbar_{f,j}-B^\hbar_{g,j}\|\prod_{\ell\neq j}\norm{u_\ell}_{L^\ii(\R^d)}\norm{v_\ell}_{L^\ii(\R^d)}.$$
Therefore we only have to prove that 
$$\|B^\hbar_{f,j}-B^\hbar_{g,j}\|\leq C\hbar^d\sqrt\hbar$$
for all $j=1,...,k$. We will actually prove that 
\begin{equation}
\|B^\hbar_{f,j}-\hbar^du_j(x)v_j(-i\hbar\nabla)\|\leq C\hbar^d\sqrt\hbar
\label{eq:change_quantization}
\end{equation}
which will of course imply the previous bound. We ignore the index $j$ for simplicity and compute the kernel 
\begin{align*}
B^\hbar_{f}(y,y')&=\frac{1}{(2\pi)^d\hbar^{\frac{d}{2}}}\int_{\R^d}\int_{\R^d} u(x)v(p)f\left(\frac{y-x}{\sqrt\hbar}\right)f\left(\frac{y'-x}{\sqrt\hbar}\right)e^{i\frac{p\cdot(y-y')}{\hbar}}\,dx\,dp\\
&=\frac{1}{(2\pi)^{\frac{d}2}\hbar^{\frac{d}{2}}} \check{v}\left(\frac{y-y'}{\hbar}\right)\int_{\R^d} u(x)f\left(\frac{y-x}{\sqrt\hbar}\right)f\left(\frac{y'-x}{\sqrt\hbar}\right)\,dx
\end{align*}
where $\check{v}$ is the inverse Fourier  transform of the function $v$.
Now the idea is to replace $u(x)$ by $u(y)$ and $f((y-x)/\hbar)$ by $f((y'-x)/\hbar)$ and to then use that $\int_{\R^d}f^2=1$. Assuming that $\check{v}$ and $f$ have their support in a ball of radius $R$, we find
\begin{multline*}
\left|B^\hbar_{f}(y,y')-(2\pi)^{-\frac{d}2}\, u(y)\,\check{v}\left(\frac{y-y'}{\hbar}\right)\right|\\
\leq \frac{R\sqrt{\hbar}}{(2\pi)^{\frac{d}2}}\left|\check{v}\left(\frac{y-y'}{\hbar}\right)\right|\left(\norm{\nabla u}_{L^\ii}+\norm{\nabla f}_{L^\ii}\norm{f}_{L^1}\norm{u}_{L^\ii}\right).
\end{multline*}
The operator with kernel $|\check{v}((y-y')/\hbar)|$ has the operator norm $\norm{\check{v}(\cdot/\hbar)}_{L^1(\R^d)}=\hbar^d\norm{\check{v}}_{L^1(\R^d)}$ and this concludes the proof of~\eqref{eq:change_quantization}, hence of the lemma.
\end{proof}

Together with the content of Section~\ref{sec:link_Wigner} dealing with the link between the Husimi and Wigner measures, this concludes the proof of the first part of Theorem~\ref{thm:subsequences_m_k}. If now the kinetic energy of $\Psi_N$ satisfies the semi-classical bound~\eqref{eq:semi-classical_bound}, then by~\eqref{eq:formula_kinetic_energy} we have 
$$(2\pi)^{-d}\int_{M}|p|^2m^{(1)}_{f,\Psi_N}(x,p)\,dx\,dp\leq C+N^{-1/d}\int_{\R^d}|\nabla f|^2.$$
Therefore the sequence is tight in the $p$ variables and the weak-limit~\eqref{eq:def_weak_limit} remains valid for functions $\phi(x_1,...,x_k,p_1,...,p_k)=U(x_1,...,x_k)$ with $U\in L^\ii_c(\R^{dk})$. By definition this is the weak convergence
$$\rho_{m^{(k)}_{f,\Psi_N}}\wto \int_\cB \rho_\mu^{\otimes k}\;d{\rm P}(\mu),$$
for every $k\geq1$. Now, $\rho^{(k)}_{\Psi_N}/N^k$ is bounded in $L^1(\R^{dk})$ and satisfies
$$\rho_{m^{(k)}_{f,\Psi_N}}=\frac{k!}{N^k}\rho^{(k)}_{\Psi_N}\ast(|f^\hbar|^2)^{\otimes k}$$
by Lemma~\ref{lem:link_PDM}. Since $|f^\hbar|^2\wto\delta$, this proves that 
$$\frac{k!}{N^k}\rho^{(k)}_{\Psi_N}\wto \int_\cB \rho_\mu^{\otimes k}\;d{\rm P}(\mu)$$
weakly-$\ast$ on $C^0_0(\R^{dk})$.

For $k=1$, the Lieb-Thirring inequality~\cite{LieThi-75,LieThi-76,LieSei-09} gives us
$$\int_{\R^d}\rho^{(1)}_{\Psi_N}(x)^{1+2/d}\,dx\leq C\pscal{\Psi_N\sum_{j=1}^N(-\Delta)_{x_j}\Psi_N}\leq CN^{1+2/d}$$
and, therefore, $\rho^{(1)}_{\Psi_N}/N$ is bounded in $L^1\cap L^{1+2/d}$. It must thus converge weakly in that space.
This concludes the proof of Theorem~\ref{thm:subsequences_m_k}.\qed

\section{Convergence of the energy: proof of Theorem~\ref{thm:CV_GS_energy}}\label{sec:proof_CV_energy}

This section is devoted to the proof of Theorem~\ref{thm:CV_GS_energy} and to the derivation of some a priori estimates that will be useful in the remaining of the article.

\subsection{Upper bound on $E(N)$ using Hartree-Fock theory}

\begin{proposition}[Upper bound]\label{lem:upper_bound}
Assume that $w$ is even, that $w,|A|^2,V_-\in L^{1+d/2}(\R^d)+L^\ii_\epsilon(\R^d)$ and that 
$$V_+\in L^1_{\rm loc}(\R^d).$$
We have
$$\limsup_{N\to\ii}\frac{E(N)}{N}\leq e^V_{\rm TF}(1).$$
\end{proposition}

\begin{proof}
Taking a Hartree-Fock state $\Psi=(N!)^{-1/2}\det(f_i(x_j))$ and introducing the corresponding density matrix
$$\gamma=\sum_{j=1}^N|f_j\rangle\langle f_j|,$$
with $\rho_\gamma:=\sum_{j=1}^N|f_j|^2$, we see that 
\begin{multline*}
E(N)\leq \inf_{\substack{\gamma^2=\gamma=\gamma^*\\ \tr\gamma=N}}\Bigg\{\tr(-i\hbar\nabla+A)^2\gamma+\int_{\R^d}V\rho_\gamma\\
+\frac{1}{2N}\iint_{\R^d\times\R^d}w(x-y)\big(\rho_\gamma(x)\rho_\gamma(y)-|\gamma(x,y)|^2\big)\,dx\,dy\Bigg\}.
\end{multline*}

We first estimate the exchange term (the last term in the previous formula). We write $w=w_1+w_2$ with $w_1\in L^{1+d/2}(\R^d)$ and $w_2\in L^\ii(\R^d)$. Then
$$\iint_{\R^d\times\R^d}|w_2(x-y)|\,|\gamma(x,y)|^2\,dx\,dy\leq \norm{w_2}_{L^\ii}\tr(\gamma^2)=\norm{w_2}_{L^\ii}N.$$
We have used here that $\tr\gamma^2=\tr\gamma=N$. For the term involving $w_1$, we use the Hölder and Gagliardo-Nirenberg-Sobolev inequalities to obtain
\begin{align*}
\int_{\R^d}|w_1(x)|\,|f(x)|^2\,dx&\leq \norm{w_1}_{L^{1+d/2}(\R^d)}\norm{f}_{L^{2+4/d}(\R^d)}^2\\
&\leq\norm{w_1}_{L^{1+d/2}(\R^d)}\left(\epsilon \norm{\nabla f}_{L^2(\R^d)}^2+\frac{C}{\epsilon^{d/2}}\norm{f}_{L^2(\R^d)}^2\right)
\end{align*}
for all $\epsilon>0$. Applying this in the variable $x$ with $y$ fixed, this gives
$$\iint_{\R^d\times\R^d}|w_1(x-y)|\,|\gamma(x,y)|^2\,dx\,dy\leq \norm{w_1}_{L^{1+d/2}(\R^d)}\left(\epsilon \tr(-\Delta)\gamma+\frac{CN}{\epsilon^{d/2}}\right).$$
Combining the two inequalities, we have shown that
$$\left|\frac{1}{2N}\iint_{\R^d\times\R^d}w(x-y)\,|\gamma(x,y)|^2\,dx\,dy\right|\leq \frac{\epsilon}{N} \tr(-\Delta)\gamma+C(1+\epsilon^{-d/2})$$
for all $\epsilon>0$. In dimensions $d\geq3$, we can simply take $\epsilon=1$. In dimensions $d=1$ and $d=2$ we take  for instance $\epsilon=N^{d-2}(\log N)^{-1}$. In all cases, we arrive at
$$\left|\frac{1}{2N}\iint_{\R^d\times\R^d}w(x-y)\,|\gamma(x,y)|^2\,dx\,dy\right|\leq \frac{\epsilon_N}{N^{2/d}} \tr(-\Delta)\gamma+N\epsilon_N$$
for some $\epsilon_N\to0$ which is independent on $\gamma$. 
Finally, in order to relate $\tr(-\Delta)\gamma$ to the magnetic kinetic energy, we may use that
\begin{equation}
|-i\hbar\nabla+A|^2=-\hbar^2\Delta-i\hbar\nabla\cdot A-i\hbar A\cdot\nabla +|A|^2\geq -\frac{\hbar^2}{2}\Delta-|A|^2 
 \label{eq:get_rid_magnetic}
\end{equation}
from which we deduce the final upper bound
\begin{multline}
\limsup_{N\to\ii}\frac{E(N)}{N}\leq \limsup_{N\to\ii}\inf_{\substack{\gamma^2=\gamma=\gamma^*\\ \tr\gamma=N}}\Bigg\{\frac{1+2\epsilon_N}{N}\tr|-i\hbar\nabla+A|^2\gamma\\
+\frac1N\int_{\R^d}\big(V+2\epsilon_N|A|^2\big)\rho_\gamma+\frac1{2N^2}\iint_{\R^d\times\R^d}w(x-y)\rho_\gamma(x)\rho_\gamma(y)\,dx\,dy\Bigg\}.\label{eq:upper_bound_almost}
\end{multline}

By using semi-classical analysis, we now construct an appropriate trial state for the variational problem on the right side, based on the classical probability density on phase space.

\begin{lemma}[Semi-classical analysis in a compact domain]\label{lem:semi-classics}
Let $\rho\geq0$ be a fixed function in $C^\ii_c(\R^d)$ with support in the cube $C_R=(-R/2,R/2)^d$, such that $\rho\geq0$ and $\int_{C_R}\rho=1$. 
Let $A \in L^{2+d}(C_R)$ be a magnetic vector potential.
Define
$$\gamma_N:=\1\left((-i \hbar\nabla + A)^2_{C_R}-c_{\rm TF}\rho(x)^{2/d}\leq 0\right)$$
where $(-i \hbar\nabla + A)^2_{C_R}$ is the magnetic Dirichlet Laplacian in the cube and $c_{\rm TF}=4\pi^2(d/|S^{d-1}|)^{2/d}$. Then we have
\begin{equation}
\lim_{N\to\ii}N^{-1}\tr(-i \hbar\nabla + A)^2\gamma_N
=\frac{d}{d+2}c_{\rm TF}\int_{\R^d}\rho(x)^{1+2/d}\,dx
\end{equation}
and 
\begin{equation}
\lim_{N\to\ii}N^{-1}\tr\gamma_N
=\int_{\R^d}\rho(x)\,dx.
\end{equation}
Furthermore, $\rho_{\gamma_N}/N\wto\rho$ weakly in $L^1(\R^d)$ and weakly-$\ast$ in $L^\ii(\R^d)$.

The same properties all hold true if $\gamma_N$ is replaced by the projection $\tilde\gamma_N$ onto the $N$ lowest eigenvectors of $(-i \hbar\nabla + A)^2_{C_R}-c_{\rm TF}\rho(x)^{2/d}$.
\end{lemma}

\begin{remark}
The lemma is true under much weaker assumptions on $A$, but we have used the same hypothesis as in the rest of the paper for convenience.
\end{remark}

The lemma is proved in Appendix~\ref{app:proof_lemma_semi-classics} and the proof of Proposition~\ref{lem:upper_bound} now follows immediately. Indeed, let $\rho\in C^\ii_c(\R^d)$ and $\tilde \gamma_N$ be as in the lemma (extended by zero outside of $C_R$). Then, by~\eqref{eq:upper_bound_almost}
\begin{align*}
&\limsup_{N\to\ii}\frac{E(N)}{N}\\
&\qquad\leq \lim_{N\to\ii}\Bigg\{\frac{1+2\epsilon_N}{N}\tr|-i\hbar\nabla+A|^2{\tilde\gamma_N}+\frac1N\int_{\R^d}(V+2\epsilon_N|A|^2)\rho_{\tilde\gamma_N}\\
&\qquad\qquad\qquad+\frac1{2N^2}\iint_{\R^d\times\R^d}w(x-y)\rho_{\tilde\gamma_N}(x)\rho_{\tilde\gamma_N}(y)\,dx\,dy\Bigg\}\\
&\qquad=\frac{d}{d+2}4\pi^2\left(\frac{d}{|S^{d-1}|}\right)^{2/d}\int_{\R^d}\rho^{1+2/d}+\int_{\R^d}V\rho\\
&\qquad\qquad\qquad+\frac1{2}\iint_{\R^d\times\R^d}w(x-y)\rho(x)\rho(y)\,dx\,dy.
\end{align*}
Here we have used that 
$$N^{-1}\int_{\R^d}V\rho_{\tilde\gamma_N}=N^{-1}\int_{C_R}V\rho_{\tilde\gamma_N}\to\int_{C_R}V\rho,$$
that $\int|A|^2\rho_{\tilde\gamma_N}$ is uniformly bounded and that 
$$N^{-2}\iint_{C_R\times C_R}w(x-y)\rho_{\tilde\gamma_N}(x)\rho_{\tilde\gamma_N}(y)\,dx\,dy\to \iint_{C_R\times C_R}w(x-y)\rho(x)\rho(y)\,dx\,dy$$
since $\rho_{\tilde\gamma_N}/N$ converges weakly in $L^1(\R^d)$ and weakly-$\ast$ in $L^\ii(\R^d)$, and $V,w,|A|^2\in L^1_{\rm loc}(\R^d)$ by assumption. 
\end{proof}

\subsection{A priori estimates}

We now derive some \emph{a priori} bounds on any (normalized) sequence $\{\Psi_N\}$ of $N$-particle states with $\pscal{\Psi_N,H_N\Psi_N}=O(N)$, which will be useful for the lower bound.

\begin{lemma}[Lieb-Thirring]\label{lem:LT}
There exists a constant $C$ such that 
\begin{equation}
H_N\geq \sum_{j=1}^N\left(-\frac{\Delta_j}{4N^{\frac{2}d}}+V_+(x_j)\right)-CN\geq -CN.
\label{eq:stability}
\end{equation}
In particular, any $\{\Psi_N\}$ such that $\pscal{\Psi_N,H_N\Psi_N}=O(N)$ must satisfy
\begin{multline}
\frac{1}{N^{1+\frac{2}d}}\pscal{\Psi_N,\left(\sum_{j=1}^N-\Delta_j\right)\Psi_N}\\+\int_{\R^d}\left(\frac{\rho^{(1)}_{\Psi_N}(x)}{N}\right)^{1+\frac2d}\,dx+\int_{\R^d}V_+\frac{\rho^{(1)}_{\Psi_N}(x)}N\,dx\leq C,
\label{eq:estim_rho}
\end{multline}
with $\rho^{(1)}_{\Psi_N}$ the one-particle densities defined in~\eqref{eq:def_rho_k}. In addition, for any potential $0\leq f=f_1+f_2\in L^{1+d/2}(\R^d)+L^\ii(\R^d)$, we have
\begin{multline}
\frac1N\int_{\R^d}f(x)\rho^{(1)}_{\Psi_N}(x)\,dx+\frac{1}{N^2}\int_{\R^d}\int_{\R^d}f(x-y)\rho^{(2)}_{\Psi_N}(x,y)\,dx\,dy\\
\leq C\left(\norm{f_1}_{L^{1+d/2}(\R^d)}+\norm{f_2}_{L^\ii(\R^d)}\right).
\label{eq:estim_interaction}
\end{multline}
\end{lemma}

\begin{proof}
First we get rid of the magnetic field at the expense of a factor in front of the kinetic energy, by using the simple bound 
$$|-i\hbar\nabla+A|^2\geq -\frac{\hbar^2}{2}\Delta-|A|^2$$
as in~\eqref{eq:get_rid_magnetic}. Our assumption on $A$ implies that $|A|^2\in L^{1+d/2}(\R^d)+L^\ii_\epsilon(\R^d)$, just as for $V_-$. We may therefore write $V_-+|A|^2=V_1+V_2\in L^{1+d/2}(\R^d)+L^\ii(\R^d)$. The estimate on $V_2$ being obvious, we use the Lieb-Thirring inequality~\cite{LieThi-75,LieThi-76,LieSei-09} to obtain
\begin{equation}
\sum_{j=1}^N\left(-\frac{\Delta_j}{8N^{\frac{2}d}}-V_1(x_j)\right)\geq -\frac14\tr\left(-\frac{\Delta_j}{N^{\frac{2}d}}-4V_1\right)_-\geq -CN\int_{\R^d}V_1^{1+d/2}.
\label{eq:Lieb-Thirring}
\end{equation}
In a similar manner, we note that $|\Psi|^2$ is symmetric and rewrite, following L\'evy-Leblond~\cite{LevyLeblond-69},
\begin{align*}
\pscal{\Psi,\left(\frac{1}{N}\sum_{1\leq k<\ell\leq N}w_-(x_k-x_\ell)\right)\Psi}&=\frac{N-1}{2}\pscal{\Psi,w_-(x_1-x_2)\Psi}\\
&=\frac{1}{2}\pscal{\Psi,\left(\sum_{2\leq j\leq N}w_-(x_1-x_j)\right)\Psi}.
\end{align*}
We now use that $\Psi(x_1,\cdot)$ is antisymmetric in the remaining $N-1$ variables and the above Lieb-Thirring estimate in those variables. Since $w_-=w_1+w_2\in L^{1+d/2}(\R^d)+L^\ii(\R^d)$, we obtain as before
\begin{equation}
\sum_{j=1}^N-\frac{\Delta_j}{8N^{\frac{2}d}}-\frac{1}{N}\sum_{1\leq k<\ell\leq N}w_-(x_k-x_\ell)\\ \geq -CN\left(\int_{\R^d}w_1^{1+d/2}+\norm{w_2}_{L^\ii}\right).
\label{eq:estim_interaction_preliminary}
\end{equation}
This concludes the proof of~\eqref{eq:stability}.
The estimate~\eqref{eq:estim_rho} follows from the Lieb-Thirring inequality written in terms of the density
\begin{equation}
\pscal{\Psi_N,\left(\sum_{j=1}^N-\Delta_j\right)\Psi_N}\geq C\int_{\R^d}\big(\rho_{\Psi_N}^{(1)}\big)^{1+2/d},
\label{eq:Lieb-Thirring_density}
\end{equation}
which is dual to~\eqref{eq:Lieb-Thirring}.
In~\eqref{eq:estim_interaction}, the one-body part follows directly from~\eqref{eq:estim_rho} and Hölder's inequality. For the two-body estimate, we use~\eqref{eq:estim_interaction_preliminary} with $w_-$ replaced by $f/\epsilon$ and obtain
$$\epsilon\sum_{j=1}^N-\frac{\Delta_j}{8N^{\frac{2}d}}-\frac{1}{N}\sum_{1\leq k<\ell\leq N}f(x_k-x_\ell)\\ \geq -CN\left(\epsilon^{-d/2}\int_{\R^d}f_1^{1+d/2}+\norm{f_2}_{L^\ii}\right).$$
Taking the expectation value in the state $\Psi_N$ and using~\eqref{eq:estim_rho}, we find~\eqref{eq:estim_interaction} after optimizing over $\epsilon$.
\end{proof}

\subsection{Lower bound and end of the proof of Theorem~\ref{thm:CV_GS_energy}}\label{sec:energy_lower_bound}
Here we establish the lower bound, which concludes the proof of Theorem~\ref{thm:CV_GS_energy}.

\begin{proposition}[Lower bound]\label{lem:lower_bound}
Under the assumptions of Theorem~\ref{thm:CV_GS_energy}, we have
\begin{equation}
\liminf_{N\to\ii}\frac{E(N)}{N}\geq e^V_{\rm TF}(1),
\label{eq:lower_bd_to_be_proven}
\end{equation}
\end{proposition}

\begin{proof}
The first step is to regularize the interaction potential $w$. For every $\epsilon>0$, we may write $w=f+g$ with $f\in L^{1+d/2}(\R^d)$ and $\|g\|_{L^\ii}\leq \epsilon$. We then consider the problem $E'(N)$ in which $w$ has been replaced by another potential $w'$, and a state $\Psi_N$ such that $\pscal{\Psi_N,H_N\Psi_N}=E(N)+o(N)$. Then by~\eqref{eq:estim_interaction}, we obtain
\begin{multline*}
\pscal{\Psi_N,\sum_{1\leq j<k\leq N}w(x_j-x_k)\Psi_N}\\ \geq \pscal{\Psi_N,\sum_{1\leq j<k\leq N}w'(x_j-x_k)\Psi_N}-CN\norm{f-w'}_{L^{1+d/2}(\R^d)}-\frac{N-1}{2}\epsilon
\end{multline*}
and, therefore,
\begin{equation}
E(N)\geq E'(N)-CN\norm{f-w'}_{L^{1+d/2}(\R^d)}-\frac{N-1}{2}\epsilon.
\label{eq:regularize}
\end{equation}
A similar bound holds for the Thomas-Fermi model and we conclude by an `$\epsilon/2$ argument' that it suffices to prove~\eqref{eq:lower_bd_to_be_proven} for a smooth interaction potential $w'$ which  approximates $f$. For simplicity of notation, we will assume for the rest of the proof that $w$ is itself smooth enough. The exact property we need is that $\widehat{w}\in L^1(\R^d)$, which implies that $w$ is a continuous bounded function that tends to zero at infinity.

Now, we write $w=w_1-w_2$ where $\widehat{w_1}=(\widehat{w})_+$ and $\widehat{w_2}=(\widehat{w})_-$ which are also in $L^1(\R^d)$. Note that $w_1$ and $w_2$ are in addition both even since $\widehat{w}$ is real. We will use the fact that
\begin{equation}
\pscal{\Psi_N,\sum_{1\leq j<k\leq N}w_{1,2}(x_j-x_k)\Psi_N}=\frac{N(N-1)}{2}\pscal{\Psi_N,w_{1,2}(x_\ell-x_m)\Psi_N}
\label{eq:trick_symmetry}
\end{equation}
for any $\ell\neq m$, by the symmetry of $|\Psi|^2$. The idea is now to split the $N$ particles into two groups of, respectively $M$ and $L$ particles, with $L=\lfloor\sqrt{N}\rfloor$. The first group will just be $x_1,...,x_M$ whereas the second (smaller) group will be denoted for convenience as $y_1=x_{M+1},...,y_{L}=x_N$. Then, we express the repulsive part $w_1$ of the potential $w$ using only the particles in the first group and write the attractive part $w_2$ as an interaction between the two groups. This means that we rewrite, using~\eqref{eq:trick_symmetry},
\begin{multline*}
\pscal{\Psi_N,\sum_{1\leq j<k\leq N}w_1(x_j-x_k)\Psi_N}\\=\frac{N(N-1)}{M(M-1)}\pscal{\Psi_N,\sum_{1\leq m<m'\leq M}w_1(x_m-x_{m'})\Psi_N} 
\end{multline*}
on one hand, and
\begin{align*}
&-\pscal{\Psi_N,\sum_{1\leq j<k\leq N}w_2(x_j-x_k)\Psi_N}\\
&\qquad\qquad=N(N-1)\pscal{\Psi_N,\left(\frac{w_2(y_1-y_2)}2-w_2(x_1-y_1)\right)\Psi_N}\\
&\qquad\qquad=\frac{N(N-1)}{L(L-1)}\pscal{\Psi_N,\sum_{1\leq \ell<\ell'\leq L}w_2(y_\ell-y_{\ell'})\Psi_N}\\
&\qquad\qquad\qquad\qquad-\frac{N(N-1)}{LM}\pscal{\Psi_N,\sum_{m=1}^M\sum_{\ell=1}^Lw_2(x_m-y_{\ell})\Psi_N}
\end{align*}
on the other hand. Expressing the one-particle terms in a similar manner using only the $M$ particles of the first group, we deduce that
$$\frac{\pscal{\Psi_N,H_N\Psi_N}}N=\frac{\pscal{\Psi_N,\tilde H\Psi_N}}M$$
where 
\begin{multline*}
\tilde H = \sum_{m=1}^M |-i\hbar\nabla_m+A(x_m)|^2+V(x_m)+\frac{1-1/N}{M-1}\sum_{1\leq m<m'\leq M}w_1(x_m-x_{m'})\\
+\frac{M(1-1/N)}{L(L-1)}\sum_{1\leq \ell<\ell'\leq L}w_2(y_\ell-y_{\ell'})-\frac{(1-1/N)}{L}\sum_{m=1}^M\sum_{\ell=1}^Lw_2(x_m-y_{\ell}).
\end{multline*}
This new Hamiltonian corresponds to a system with $M\sim N$ quantum particles which repel through the potential $w_1$ and $L\sim\sqrt N$ classical particles that repel through the potential $w_2$, with an additional attraction between the two species given by $-w_2$. In other words, we have transformed the attractive part $w_2$ of $w$ into an interaction with an auxiliary system of $L$ particles that repel each other. A similar approach was used by Lieb-Thirring~\cite{LieThi-84} and Lieb-Yau~\cite{LieYau-87} and it was in turm inspired by arguments of Levy-Leblond~\cite{LevyLeblond-69} and Dyson-Lenard~\cite{DysLen-67}. For bosons the argument was explained in~\cite{Lewin-15} (there one can take $M$ and $L$ of the order $N$, which simplifies a bit the argument). 
From the antisymmetry of $\Psi_N$ in the $M$ first variables, we conclude that
$$\frac{\pscal{\Psi_N,H_N\Psi_N}}N\geq \inf_{y_1,...,y_L}\frac{\inf\sigma_{\bigwedge^M_1L^2(\R^d)}(\tilde H)}M$$
and it therefore remains to estimate the bottom of the spectrum of $\tilde H$ for $M$ fermions, uniformly with respect to the positions $y_1,...,y_L$ of the $L$ classical particles. The rest of the argument will then be based on the following well-known lemma, which allows to bound from below the two-body interaction by a one-particle term.

\begin{lemma}[Estimating the two-body potential by a one-particle term]\label{lem:Lieb-Oxford_easy}
Let $f$ be a function such that $\widehat{f}$ is in $L^1(\R^d)$ and non-negative, and denote
$$D_f(g,g):=\int_{\R^d}\int_{\R^d}g(x)g(y)f(x-y)\,dx\,dy=(2\pi)^{\frac{d}2}\int_{\R^d}|\widehat{g}(k)|^2\widehat{f}(k)\,dk\geq0.$$
Then we have
\begin{equation}
\sum_{1\leq k<k'\leq K}f(z_k-z_{k'})\geq -\frac{f(0)}{2}K+\sum_{k=1}^Kf\ast\eta(z_k)-\frac12D_f(\eta,\eta)
\label{eq:relation_positive_Fourier}
\end{equation}
for every $\eta\in L^1(\R^d)$ and every $z_1,...,z_K\in\R^d$, and
\begin{equation}
\pscal{\tilde\Psi,\sum_{1\leq k<k'\leq K}f(z_k-z_{k'})\tilde\Psi}\geq -\frac{f(0)}{2}K+\frac12D_f(\rho^{(1)}_{\tilde\Psi},\rho^{(1)}_{\tilde\Psi}).
\label{eq:Lieb-Oxford-easy}
\end{equation}
for every $K$-particle state $\tilde\Psi$.
\end{lemma}

\begin{proof}[Proof of Lemma~\ref{lem:Lieb-Oxford_easy}]
Use that $D_f(g,g)\geq0$ for $g=\sum_{k=1}^K\delta_{z_k}-\eta$ and then take $\eta=\rho_{\tilde\Psi}^{(1)}$.
\end{proof}

Now, let $\tilde\Psi$ be any fermionic $M$-particle state. Using~\eqref{eq:Lieb-Oxford-easy} with $f=w_1$, we get
\begin{align*}
&\frac{1-1/N}{M-1}\pscal{\tilde\Psi,\sum_{1\leq m<m'\leq M}w_1(x_m-x_{m'})\tilde\Psi}\\
&\qquad\qquad \geq \frac{1-1/N}{2(M-1)}D_{w_1}(\rho^{(1)}_{\tilde\Psi},\rho^{(1)}_{\tilde\Psi})-\frac{M(1-1/N)}{2(M-1)}w_1(0)\\
&\qquad\qquad\geq \frac{1}{2M}D_{w_1}(\rho^{(1)}_{\tilde\Psi},\rho^{(1)}_{\tilde\Psi})-\frac{M}{2(M-1)}w_1(0).
\end{align*}
Using now~\eqref{eq:relation_positive_Fourier} with $f=w_2$ and $\eta=\rho_{\tilde\Psi}^{(1)}(L-1)/M$, we find 
\begin{align*}
&\frac{M(1-1/N)}{L(L-1)}\sum_{1\leq \ell<\ell'\leq L}w_2(y_\ell-y_{\ell'})\\
&\qquad\qquad\qquad-\frac{(1-1/N)}{L}\pscal{\tilde\Psi,\sum_{m=1}^M\sum_{\ell=1}^Lw_2(x_m-y_{\ell})\tilde\Psi}\\
&\qquad\geq -\frac{(L-1)(1-1/N)}{2ML}D_{w_2}(\rho^{(1)}_{\tilde\Psi},\rho^{(1)}_{\tilde\Psi})-\frac{M(1-1/N)}{2(L-1)}w_2(0)\\
&\qquad\geq -\frac{1}{2M}D_{w_2}(\rho^{(1)}_{\tilde\Psi},\rho^{(1)}_{\tilde\Psi})-\frac{M}{2(L-1)}w_2(0).
\end{align*}
Recall that $\tilde\Psi$ only depends on the $x_m$'s. We have therefore proved the following lower bound, independent of the $y_\ell$'s,
\begin{align*}
\frac{\langle\tilde\Psi,\tilde H\tilde\Psi\rangle}M \geq& \frac1M\tr\left[\big(|-i\hbar\nabla+A|^2+V\big)\gamma_{\tilde\Psi}^{(1)}\right]+\frac{1}{2M^2}D_{w}(\rho^{(1)}_{\tilde\Psi},\rho^{(1)}_{\tilde\Psi})\\
&\qquad-\frac{1}{2(M-1)}w_1(0)-\frac{1}{2(L-1)}w_2(0)\\
\geq&\frac1M\tr\big(|-i\hbar\nabla+A|^2+V\big)\gamma_{\tilde\Psi}^{(1)}+\frac{1}{2M^2}D_{w}(\rho^{(1)}_{\tilde\Psi},\rho^{(1)}_{\tilde\Psi})\\
&\qquad-\frac{\|\widehat w\|_{L^1}}{(2\pi)^{d/2}\sqrt{N}},
\end{align*}
where we recall that $M=N-\lfloor\sqrt{N}\rfloor$ and that $\hbar=N^{-1/d}\simeq M^{-1/d}$. The term on the right is (up to the last constant) the reduced Hartree-Fock (rHF) energy of the $M$ particles, that is well-known to converge to the Thomas-Fermi problem in the limit $M\simeq N\to\ii$~\cite{LieSim-73,LieSim-77,Lieb-81b}. Even if we only need the lower bound, we state and prove the full convergence.

\begin{lemma}[Reduced-Hartree-Fock model]\label{lem:rHF}
Under the same assumptions on $V$, $A$ and $w$ as in Theorem~\ref{thm:CV_GS_energy}, we have
\begin{equation}
\lim_{\substack{N\to\ii\\ \hbar N^{1/d}\to1}}\inf_{\substack{0\leq\gamma\leq 1\\ \tr(\gamma)=N}}\left\{\frac1N\tr\big(|-i\hbar^2\nabla+A|^2+V\big)\gamma+\frac{1}{2N^2}D_{w}(\rho_\gamma,\rho_\gamma)\right\}=e_{\rm TF}^V(1).
\end{equation}
\end{lemma}

\begin{proof}[Proof of Lemma~\ref{lem:rHF}]
The upper bound follows from Lemma~\ref{lem:semi-classics}. By~\eqref{eq:get_rid_magnetic} and the Lieb-Thirring inequality expressed in terms of the one-particle density matrix
$$\tr(-\Delta)\gamma\geq C\int_{\R^d}\rho_\gamma^{1+\frac{2}{d}},\qquad \forall 0\leq \gamma\leq 1$$
we deduce that any appropriate minimizing sequence $\gamma_N$ satisfies
\begin{equation}
N^{2/d}\int_{\R^d} V_+\rho_{\gamma_N}+\int_{\R^d}\rho_{\gamma_N}^{1+2/d}+\tr(-\Delta)\gamma_N\leq CN^{1+2/d}. 
 \label{eq:bound_rHF}
\end{equation}
Under the assumption that $V\in L^{1+d/2}(\R^d)+L^\ii_\eps(\R^d)$, these estimates can be used to replace the potentials $V$, $A$ and $w$ by functions in $C^\ii_c(\R^d)$.
When $V_+\to+\ii$ at infinity, we first bound from below $V_+$ by the truncated potential 
$$V_+^M:=V_+\1(|V_+|\leq M)+M\1(|V_+|\geq M)$$
and take $M\to\ii$ at the very end of the argument, using that 
$$\lim_{M\to\ii}e^{V_+^M-V_-}_{\rm TF}(1)=e^{V}_{\rm TF}(1).$$ 
The previous estimates then allow to replace $A$ and $w$ by functions in $C^\ii_c(\R^d)$ and $V_M$ by a function in $C^\ii(\R^d)$, which is equal to $M$ outside of a sufficiently large ball.

We now fix a real-valued normalized function $f\in L^2(\R^d)$ and use formulas~\eqref{eq:formula_kinetic_energy_magnetic_field} and~\eqref{eq:link_PDM_x} to obtain
\begin{align*}
\tr(|-i\hbar^2\nabla+A|^2+V)\gamma_N=&\frac{1}{(2\pi\hbar)^d}\int_{\R^{2d}}\big(|p+A(x)|^2+V(x)\big)m_N(x,p)\,dx\,dp\\
&-N\hbar\int_{\R^d}|\nabla f|^2-2\tr(A-A\ast|f^\hbar|^2)\cdot(-i\hbar\nabla)\gamma_N\\
&+\int_{\R^d}(|A|^2-|A|^2\ast|f^\hbar|^2+V-V\ast|f^\hbar|^2)\rho_{\gamma_N}
\end{align*}
and
\begin{align*}
D_{w}(\rho_{\gamma_N},\rho_{\gamma_N})=\hbar ^{-2d}D_{w}(\rho_{m_N},\rho_{m_N})+D_{w-w\ast|f^\hbar|^2\ast|f^\hbar|^2}(\rho_{\gamma_N},\rho_{\gamma_N})
\end{align*}
where $m_N(x,p)=\pscal{f_{x,p}^\hbar,\gamma_N f_{x,p}^\hbar}$ and $\rho_{m_N}=(2\pi)^{-d}\int_{\R^d}m_N(x,p)\,dp$. Note that $0\leq m_N\leq 1$ and $(2\pi)^{-d}\iint m_N(x,p)\,dx\,dp=1$, hence $m_N$ is a suitable trial function for the Vlasov problem. Using the bound~\eqref{eq:bound_rHF} we have
\begin{align*}
\tr(A-A\ast|f^\hbar|^2)\cdot(-i\hbar\nabla)\gamma&\geq -\norm{(-i\hbar\nabla)\sqrt{\gamma_N}}_{\gS_2}\norm{\sqrt{\gamma_N}(A-A\ast|f^\hbar|^2)}_{\gS^2}\\
&=-\sqrt{\hbar^2\tr(-\Delta)\gamma_N}\sqrt{\int_{\R^d}(A-A\ast|f^\hbar|^2)^2\rho_{\gamma_N}}\\
&\geq -CN\norm{A-A\ast|f^\hbar|^2}_{L^{2+d}(\R^d)}.
\end{align*}
We can argue similarly for the other terms (when $V\in C^\ii(\R^d)$ is equal to $M$ outside of a large ball $B_R$, we have to use that $V-V\ast|f^\hbar|^2=(V-M)-(V-M)\ast|f^\hbar|^2$ has compact support and converges to $0$ strongly in $L^{1+d/2}(\R^d)$). We obtain
\begin{multline*}
 \frac1N \tr(|-i\hbar^2\nabla+A|^2+V)\gamma_N+\frac{1}{2N^2}D_{w}(\rho_{\gamma_N},\rho_{\gamma_N})\\
 \geq \frac{1}{(2\pi)^dN\hbar^d}\int_{\R^{2d}}\big(|p+A(x)|^2+V(x)\big)m_N(x,p)\,dx\,dp\\
 +\frac{1}{2N^2\hbar^{2d}}D_{w}(\rho_{m_N},\rho_{m_N})+o(1).
\end{multline*}
Since $N\hbar^d\to1$, we obtain the Vlasov energy on the right and the proof of Lemma~\ref{lem:rHF} is complete.
\end{proof}

This now concludes the proof of the lower bound~\eqref{eq:lower_bd_to_be_proven}, hence of Theorem~\ref{thm:CV_GS_energy}.
\end{proof}

\section{Convergence of states: proof of Theorems~\ref{thm:CV_states_confined} and~\ref{thm:CV_states_unconfined}}\label{sec:proof_CV_states}

Let $\{\Psi_N\}$ be as in the statement, that is, such that $\pscal{\Psi_N,H_N\Psi_N}=E(N)+o(N)$. By the estimates in Lemma~\ref{lem:LT}, we know that $\Psi_N$ satisfies the semi-classical kinetic energy bound
\begin{equation}
\pscal{\Psi_N,\left(\sum_{j=1}^N-\Delta_j\right)\Psi_N}\leq CN^{1+2/d}.
\label{eq:semi-classical_bound2}
\end{equation}
This implies by Theorem~\ref{thm:subsequences_m_k} that, after extraction of a (not displayed) subsequence,
\begin{equation}
m^{(k)}_{f,\Psi_{N}}\wto\int_\cB \mu^{\otimes k}\;d{\rm P}(\mu),\qquad {\mathscr W}^{(k)}_{\Psi_{N}}\wto\int_\cB \mu^{\otimes k}\;d{\rm P}(\mu),
\label{eq:def_weak_limit2}
\end{equation}
\begin{equation}
\frac{k!}{N^{k}}\rho^{(k)}_{\Psi_N}\wto \int_\cB \rho_\mu^{\otimes k}\;d{\rm P}(\mu),
\label{eq:def_weak_limit_density2}
\end{equation}
weakly for all $k\geq1$, where $\rm P$ is a probability measure on 
$$\cB=\left\{\mu\in L^1(\R^{2d})\ :\ 0\leq \mu\leq 1,\ (2\pi)^{-d}\int_{\R^{2d}} \mu\leq 1\right\}.$$

Therefore, there only remains to show that $\rm P$ has its support in the set of minimizers of the Vlasov energy (in the confined case) and in the set of weak limits of minimizing sequences (in the unconfined case). 

\subsection{Confined case: proof of Theorem~\ref{thm:CV_states_confined}}

We start with the much simpler confined case where $V_+\to+\ii$ at infinity. In that case, the argument below actually gives directly the energy lower bound, and the arguments of Section~\ref{sec:energy_lower_bound} are not needed. From the bound~\eqref{eq:estim_rho} we have 
\begin{multline*}
\frac{1}{N}\int_{\R^d} V_+(x)\rho^{(1)}_{\Psi_N}(x)\,dx\\=\frac{1}{N(N-1)}\iint_{\R^{2d}} \big(V_+(x)+V_+(y)\big) \rho^{(2)}_{\Psi_N}(x,y)\,dx\,dy\leq C 
\end{multline*}
and infer that $\rho^{(1)}_{\Psi_N}/N$ and $\rho^{(2)}_{\Psi_N}/N^2$ are tight at infinity. From this we infer that there is no loss of mass at infinity and that $\rm P$ has its support on
$$\cS=\left\{\mu\in L^1(\R^{2d})\ :\ 0\leq \mu\leq 1,\ \frac{1}{(2\pi)^{d}}\int_{\R^{2d}} \mu=\int_{\R^d}\rho_\mu= 1\right\}.$$
Using~\eqref{eq:def_weak_limit_density2} and the tightness property at infinity, we are able to pass to the limit in the interaction term:
\begin{multline*}
\lim_{N\to\ii}\pscal{\Psi_N,\frac{1}{N^2}\sum_{1\leq j<k\leq N}w(x_j-x_k)\Psi_N}\\
=\lim_{N\to\ii}\int_{\R^{2d}}w(x-y)\frac{\rho_{\Psi_N}^{(2)}(x,y)}{N^2}\,dx\,dy  = \frac12\int_\cS D(\rho_\mu,\rho_\mu)\;d{\rm P}(\mu).
\end{multline*}
In order to pass from the second to the third line, one must first use the bound~\eqref{eq:estim_interaction} to replace the interaction potential $w$ by a bounded function, and then use the weak convergence of $\rho^{(2)}_{\Psi_N}$ in $L^1((\R^d)^2)$.
By Fatou's lemma for $V_+$ and the weak convergence of $\rho^{(1)}_{\Psi_N}$ for $V_-$, we also deduce that
$$\liminf_{N\to\ii}\pscal{\Psi_N,\frac{1}{N}\sum_{j=1}^NV(x_j)\Psi_N}=\liminf_{N\to\ii}\int_{\R^d}V\frac{\rho^{(1)}_{\Psi_N}}{N}\geq \int_\cS \left(\int_{\R^d}V\rho_\mu\right)\;d{\rm P}(\mu).$$
Finally, using the same argument as in the proof of Lemma~\ref{lem:rHF} for the magnetic kinetic energy, we also deduce from Fatou's lemma that
\begin{align*}
&\liminf_{N\to\ii}\pscal{\Psi_N,\frac{1}{N}\sum_{j=1}^N|-i\hbar\nabla_j+A(x_j)|^2\Psi_N}\\
&\qquad\qquad= \liminf_{N\to\ii} \frac{1}{(2\pi)^d}\iint_{\R^{2d}}|p+A(x)|^2m^{(1)}_{f,\Psi_N}(x,p)\,dx\,dp\\
&\qquad\qquad\geq \int_\cS\left(\frac{1}{(2\pi)^d}\iint_{\R^{2d}}|p+A(x)|^2\mu(x,p)\,dx\,dp\right)\,d{\rm P}(\mu).
\end{align*}
Therefore, we have proved that
$$e^V_{\rm TF}(1)\geq \lim_{N\to\ii}\frac{E(N)}{N}=\lim_{N\to\ii}\frac{\pscal{\Psi_N,H_N\Psi_N}}{N}\geq \int_\cS \cE_{\rm Vla}^{V,A}(\mu)\,d{\rm P}(\mu).$$
The minimum of the Vlasov energy is precisely $e^V_{\rm TF}(1)$ and since $\rm P$ is a probability, we have
$$\int_\cS \cE_{\rm Vla}^{V,A}(\mu)\,d{\rm P}(\mu)=\int_\cS \underbrace{\left(\cE_{\rm Vla}^{V,A}(\mu)-e_{\rm TF}^V(1)\right)}_{\geq0}\,d{\rm P}(\mu)+e_{\rm TF}^V(1)$$
and there must be equality everywhere. This implies that $\rm P$ is supported on the minimizers of $\cE_{\rm Vla}^{V,A}$. These are all of the form $m_\rho(x,p)=\1(|p+A(x)|^2\leq c_{\rm TF}\rho(x)^{2/d})$ where $\rho$ minimizes the Thomas-Fermi energy. Hence, $\rm P$ induces a probability density $\mathscr P(\rho)$ on the set of minimizers of the Thomas-Fermi energy and this concludes the proof of Theorem~\ref{thm:CV_states_confined}.\qed

\subsection{Unconfined case: proof of Theorem~\ref{thm:CV_states_unconfined}}

The proof in the unconfined situation is much more subtle and we will follow here ideas from~\cite{LewNamRou-14}, based on localization methods in Fock space~\cite{Lewin-11}.

Let $\chi$ be a smooth function on $\R^+$ with $0\leq\chi\leq1$, $\chi_{|[0,1]}\equiv1$ and $\chi_{|[2,\ii)}\equiv 0$. Denote $\chi_R(x)=\chi(|x|/R)$ and $\eta_R=(1-\chi_R^2)^{1/2}$. Let 
$$G_{R,n}^{-}={N\choose n}(\chi_R)^{\otimes n}\tr_{n+1,...,N}\Big((\eta_R)^{\otimes N-n}|\Psi_N\rangle\langle\Psi_N|(\eta_R)^{\otimes N-n}\Big)(\chi_R)^{\otimes n}$$
be the projection onto the $n$-particle space of the many-particle state localized using $\chi_R$~\cite{Lewin-11}, and, in a similar fashion,
$$G_{R,n}^{+}={N\choose n}(\eta_R)^{\otimes n}\tr_{n+1,...,N}\Big((\chi_R)^{\otimes N-n}|\Psi_N\rangle\langle\Psi_N|(\chi_R)^{\otimes N-n}\Big)(\eta_R)^{\otimes n}.$$
We remark that
\begin{equation}
\tr G_{R,n}^{+}=\tr G_{R,N-n}^{-}
\label{eq:fundamental}
\end{equation}
and that the $k$-particle density matrices are localized in the usual sense:
\begin{equation}
(\chi_R)^{\otimes k}\gamma_{\Psi_N}^{(k)}(\chi_R)^{\otimes k}=\sum_{n=k}^N\gamma_{G_{n,R}^-}^{(k)},\qquad (\eta_R)^{\otimes k}\gamma_{\Psi_N}^{(k)}(\eta_R)^{\otimes k}=\sum_{n=k}^N\gamma_{G_{n,R}^+}^{(k)}.
\label{eq:DM_localized}
\end{equation}
Next we split the $N$-body quantum energy using a partition of unity in the one-particle state as follows
\begin{align*}
\frac{\pscal{\Psi_N,H_N\Psi_N}}N&=N^{-1}\tr|-i\hbar\nabla+A|^2\gamma^{(1)}_{\Psi_N} +N^{-1}\int\rho_{\Psi_N}^{(1)}V\\
&\quad +N^{-2}\int_{\R^d}\int_{\R^d}w(x-y)\rho_{\Psi_N}^{(2)}(x,y)\,dx\,dy\\
&=N^{-1}\tr|-i\hbar\nabla+A|^2\chi_R\gamma^{(1)}_{\Psi_N}\chi_R +N^{-1}\int\rho_{\Psi_N}^{(1)}\chi_R^2V\\
&\quad +N^{-2}\int_{\R^d}\int_{\R^d}w(x-y)\chi_R^2(x)\chi_R^2(y)\rho_{\Psi_N}^{(2)}(x,y)\,dx\,dy\\
&\quad +N^{-1}\tr|-i\hbar\nabla+A|^2\eta_R\gamma^{(1)}_{\Psi_N}\eta_R\\
&\quad +N^{-2}\int_{\R^d}\int_{\R^d}w(x-y)\eta_R^2(x)\eta_R^2(y)\rho_{\Psi_N}^{(2)}(x,y)\,dx\,dy+\epsilon_{N,R},
\end{align*}
with the localization error
\begin{multline*}
\epsilon_{N,R}:=\frac{1}{R^2N^{1+2/d}}\int \chi'(|x|/R)^2\rho_{\Psi_N^{(1)}}(x)\,dx+\frac{1}{N}\int\rho_{\Psi_N}^{(1)}V\eta_R^2\\
+\frac{2}{N^2}\int_{\R^d}\int_{\R^d}w(x-y)\chi_R^2(x)\eta_R^2(y)\rho_{\Psi_N}^{(2)}(x,y)\,dx\,dy.
\end{multline*}
From the known weak convergence of the densities one has 
$$0=\lim_{R\to\ii}\left(\limsup_{N\to\ii}|\epsilon_{N,R}|\right):=\lim_{R\to\ii}\epsilon_R$$
(the term involving $w$ is treated by first replacing $w$ by a bounded function of compact support, using~\eqref{eq:estim_interaction}, and then by using the weak convergence of $\rho^{(2)}_{\Psi_N}$).
On the other hand, the terms involving $\eta_R$ may be rewritten using the property~\eqref{eq:DM_localized} as
\begin{multline*}
\frac{1}{N^{1+2/d}}\tr(-\Delta)\eta_R\gamma^{(1)}_{\Psi_N}\eta_R+\frac{1}{N^{2}}\int_{\R^d}\int_{\R^d}w(x-y)\eta_R^2(x)\eta_R^2(y)\rho_{\Psi_N}^{(2)}(x,y)\,dx\,dy\\
=\sum_{n=1}^N\frac{n}{N}\tr\left(\frac{H_{N,n}^0}{n}G_{R,n}^+\right)\geq \sum_{n=1}^N\frac{n}{N}\frac{E^0(n,n/N)}{n}\tr(G_{R,n}^+),
\end{multline*}
where $H^0_{N,n}$ is the $n$-body Hamiltonian
$$H^0_{N,n}:=\sum_{j=1}^n|-iN^{-1/d}\nabla_j+A(x_j)|^2+\frac{1}{N}\sum_{1\leq j<k\leq n}w(x_j-x_k)$$
with ground state energy $E^0(n,n/N)$. Passing to the weak limit for the terms involving $\chi_R$, we obtain as before
\begin{multline}
\lim_{N\to\ii}\frac{\pscal{\Psi_N,H_N\Psi_N}}N\geq \int_{\cB}\cE^{V,A}_{\rm Vla}\big(\chi_R^2(x)m\big)\,d{\rm P}(m)\\
+\liminf_{N\to\ii} \sum_{n=1}^N\frac{n}{N}\frac{E^0(n,n/N)}{n}\tr(G_{R,n}^+)+\epsilon_R.
\end{multline}
When $n\to\ii$ and $n/N\to\lambda>0$, then a simple adaptation of Theorem~\ref{thm:CV_GS_energy} states that
\begin{align*}
\lim_{\substack{n\to\ii\\ n/N\to\lambda}}\frac{n}{N}\frac{E^0(n,n/N)}{n}&=\lambda\inf_{\substack{\rho\geq0\\ \int_{\R^d}\rho=1}}\bigg\{\frac{d}{d+2}\lambda^{2/d}c_{\rm TF}\int_{\R^d}\rho^{1+2/d}\\
&\qquad\qquad\qquad\qquad+\frac{\lambda}{2}\int_{\R^d}\int_{\R^d}w(x-y)\rho(x)\rho(y)\,dx\,dy\bigg\}\\
&=e^0_{\rm TF}(\lambda).
\end{align*}
Arguing as in~\cite[p.~613--614]{LewNamRou-14}, it is not difficult to see that
\begin{equation*}
\lim_{N\to\ii} \sum_{n=1}^N\left(\frac{n}{N}\frac{E(n,n/N)}{n}-e^0_{\rm TF}(n/N)\right)\tr(G_{R,n}^+)=0.
\end{equation*}
Finally, using~\eqref{eq:fundamental} we find that
$$\sum_{n=1}^Ne^0_{\rm TF}\left(\frac{n}N\right)\tr(G_{R,n}^+)=\sum_{n=0}^{N-1}e^0_{\rm TF}\left(1-\frac{n}N\right)\tr(G_{R,n}^-).$$
Now the proof can be concluded based on the following Lemma, inspired of~\cite[Thm.~2.6]{LewNamRou-14}.

\begin{lemma}\label{lem:limit_particle_number}
For every continuous function $f$ on $[0,1]$, we have
\begin{equation}
\lim_{N\to\ii} \sum_{n=1}^Nf\left(\frac{n}{N}\right)\tr(G_{R,n}^-)=\int_\cB f\left(\int_{\R^d}\chi^2_R\,\rho_m\right)\,d{\rm P}(m).
\end{equation}
\end{lemma}

Indeed, the lemma implies that
$$\lim_{N\to\ii}\sum_{n=0}^{N-1}e^0_{\rm TF}\left(1-\frac{n}N\right)\tr(G_{R,n}^-)=\int_\cB e^0_{\rm TF}\left(1-\int_{\R^d}\chi^2_R\,\rho_m\right)\,d{\rm P}(m)$$
and therefore we have proved that
\begin{multline*}
\lim_{N\to\ii}\frac{\pscal{\Psi_N,H_N\Psi_N}}N\\ \geq \int_{\cB}\left\{\cE^{V,A}_{\rm Vla}\big(\chi_R^2(x)m\big)+e^0_{\rm TF}\left(1-\int_{\R^d}\chi^2_R\,\rho_m\right)\right\}\,d{\rm P}(m)+\epsilon_R.
\end{multline*}
Taking finally the limit $R\to\ii$ gives the final estimate
\begin{equation}
e^V_{\rm TF}(1)=\lim_{N\to\ii}\frac{\pscal{\Psi_N,H_N\Psi_N}}N\geq \int_{\cB}\left\{\cE^{V,A}_{\rm Vla}(m)+e^0_{\rm TF}\left(1-\int_{\R^d}\rho_m\right)\right\}\,d{\rm P}(m). 
\end{equation}
Since
$\cE^{V,A}_{\rm Vla}(m)\geq e^{V}_{\rm TF}\left(\int_{\R^d}\rho_m\right)$
and since
$e^V_{\rm TF}(1)\leq e^V_{\rm TF}(\lambda)+e^V_{\rm TF}(1-\lambda)$
for every $0\leq\lambda\leq1$, we conclude that $\rm P$ must be supported on the set 
\begin{multline*}
\bigg\{0\leq m\leq 1\ :\ \int_{\R^d}\rho_m\leq 1,\\
\cE^{V,A}_{\rm Vla}(m)=e_{\rm TF}^V\Big(\int_{\R^d}\rho_m\Big)=e_{\rm TF}^V(1)-e_{\rm TF}^0\Big(1-\int_{\R^d}\rho_m\Big)\bigg\} 
\end{multline*}
which can easily be seen to be the set of weak limits of minimizing sequences for the variational problem $e^{V,A}_{\rm Vla}(1)$. The probability measure $\rm P$ induces a probability measure $\mathscr{P}$ on the set $\cM$ of the corresponding $\rho$'s and the proof of Theorem~\ref{thm:CV_states_unconfined} is finished.\qed


\begin{proof}[Proof of Lemma~\ref{lem:limit_particle_number}]
The proof is the same as~\cite[Thm.~2.6]{LewNamRou-14}. For $k\geq1$, we have by~\eqref{eq:DM_localized}
\begin{align*}
\frac{1}{N^k}\sum_{n=k}^N{n\choose k}\tr(G_{R,n}^-)&=\frac{1}{N^k}\tr(\chi_R)^{\otimes k}\gamma_{\Psi_N}^{(k)}(\chi_R)^{\otimes k}\\
&=\frac{1}{N^k}\int_{\R^{dk}}\prod_{j=1}^k \chi_R(x_j)^2\rho_{\Psi_N}^{(k)}(x_1,...,x_k)\,dx_1\cdots dx_k\\
&\underset{N\to\ii}{\longrightarrow}\frac{1}{k!}\int_{\cB}\left(\int_{\R^d}\chi_R^2\rho_m\right)^k\,d{\rm P}(m),
\end{align*}
from the local convergence of $\rho_{\Psi_N}^{(k)}/N^k$. Now 
$$
\frac{k!}{N^k}\sum_{n=k}^N{n\choose k}\tr(G_{R,n}^-)=\sum_{n=k}^N\frac{n(n-1)\cdots (n-k+1)}{N^k}\tr(G_{R,n}^-)
$$
has the same limit as 
$\sum_{n=1}^N(n/N)^k\tr(G_{R,n}^-).$
Therefore we have proved the lemma for $f(x)=x^k$, for all $k\geq1$. The result for any continuous function $f$ follows from the density of polynomials on $[0,1]$, together with the fact that 
$\sum_{n=0}^N\tr(G_{R,n}^-)=1$
for every $N$, by definition of the localized state.
\end{proof}

\appendix
\section{Proof of Lemma~\ref{lem:semi-classics}}\label{app:proof_lemma_semi-classics}

We first get a uniform estimate on $\rho_{\gamma_N}$. Indeed, using that
$$\gamma_N=\1((-i N^{-1/d}\nabla + A)^2_{C_R}-c_{\rm TF}\rho\leq0)\leq e^{-\beta\big((-i N^{-1/d}\nabla + A)^2_{C_R}-c_{\rm TF}\rho\big)},$$
for any $\beta>0$, we deduce that
$$\rho_{\gamma_N}(x)\leq |e^{-\beta((-i N^{-1/d}\nabla + A)^2_{C_R}-c_{\rm TF}\rho)}(x,x)|.$$
From the Feynman-Kac formula and the diamagnetic inequality, we have
$$|e^{-\beta\big((-i N^{-1/d}\nabla + A)^2_{C_R}-c_{\rm TF}\rho\big)}(x,y)|\leq e^{c_{\rm TF}\beta\norm{\rho}_{L^\ii}}e^{N^{-2/d}\beta\Delta_{C_R}}(x,y),$$
where $\Delta_{C_R}$ is the (non-magnetic) Dirichlet Laplacian on $C_R$,
and hence
\begin{align*}
\rho_{\gamma_N}(x)&\leq \left(\frac{2}{R}\right)^{d}e^{c_{\rm TF}\beta\norm{\rho}_{L^\ii}}\!\!\!\!\sum_{k_1,...,k_d\in (\pi/R)\N\setminus\{0\}} e^{-N^{-2/d}\beta\sum_{j=1}^d|k_j|^2}\times\\
&\qquad\qquad\qquad\qquad\qquad\qquad\qquad\times\prod_{j=1}^d\sin^2\left(k_j (x_j-R/2)\right)\\
&\leq \left(\frac{2}{R}\right)^{d}e^{c_{\rm TF}\beta\norm{\rho}_{L^\ii}}\left(\sum_{k\in \frac{\pi}{RN^{1/d}}\N\setminus\{0\}} e^{-\beta|k|^2}\right)^d\\
&\leq e^{c_{\rm TF}\beta\norm{\rho}_{L^\ii}}N\pi^{-d}\beta^{-d/2}\left(\int_\R e^{-|k|^2}\,dk\right)^d.
\end{align*}
Optimizing over $\beta$ gives
\begin{equation}
\rho_{\gamma_N}(x)\leq \left(\frac{2ec\pi}{\pi d}\right)^{d/2}N\norm{\rho}_{L^\ii(\R^d)}^{d/2}.
\end{equation}

We have therefore shown that $\rho_{\gamma_N}/N$ is bounded in $L^\ii(C_R)$. Up to extraction of a subsequence, we may assume that $\rho_{\gamma_N}/N\wto\tilde\rho$ weakly. Now we use that
\begin{align*}
0&\geq -\frac1N\tr\big((-i N^{-1/d}\nabla + A)^2_{C_R}-c_{\rm TF}\rho\big)_-\\
&=\frac1N\tr((-i N^{-1/d}\nabla + A)^2-c_{\rm TF}\rho)\gamma_N\\
&=\frac1{N}\tr(-i N^{-1/d}\nabla + A)^2_{C_R}\gamma_N-c_{\rm TF}\frac1N\int_{C_R}\rho\rho_{\gamma_N}
\end{align*}
from which we deduce that
$$\frac1{N}\tr(-i N^{-1/d}\nabla + A)^2_{C_R}\gamma_N\leq c_{\rm TF}\frac1N\int_{C_R}\rho\rho_{\gamma_N}\leq C\norm{\rho}_{L^1}\norm{\rho}_{L^\ii}^{d/2},$$
due to the uniform upper bound on $\rho_{\gamma_N}/N$. We then introduce the semi-classical measure $m_N(x,p)=\pscal{f^\hbar_{x,p},\gamma_Nf^\hbar_{x,p}}$ and call its weak limit $m$. Arguing as before, we obtain
\begin{multline*}
\liminf_{N\to\ii}\left(\frac1{N}\tr(-i N^{-1/d}\nabla + A)^2_{C_R}\gamma_N-c_{\rm TF}\frac1N\int_{C_R}\rho\rho_{\gamma_N}\right)\\
\geq \frac1{(2\pi)^d}\int_{\R^d}\int_{\R^d}\big(|p+A(x)|^2-c_{\rm TF}\rho(x)\big)m(x,p)\,dp\,dx. 
\end{multline*}
We now use the well-known Weyl asymptotics for the energy, i.e. 
\begin{align}
\label{eq:Magn-Weyl}
&\lim_{N\to\ii} -\frac1N\tr\big((-i N^{-1/d}\nabla + A)^2_{C_R}-c_{\rm TF}\rho(x)\big)_- \nonumber\\
&\qquad\qquad=-\frac{1}{(2\pi)^d}\int_{\R^d}\int_{\R^d}\big(|p+A(x)|^2-c_{\rm TF}\rho(x)\big)_-\,dx\,dp.
\end{align}
The result \eqref{eq:Magn-Weyl} is standard for smooth vector potentials $A$. Let $A_{\epsilon}$ be a smooth approximation of $A$ in $L^2$.
Using the inequality
\begin{align*}
(1-\delta) (-i N^{-1/d}\nabla + A_{\epsilon})^2 & - \delta^{-2} |A - A_{\epsilon}|^2\leq 
(-i N^{-1/d}\nabla + A)^2 \nonumber \\
&\leq
(1+\delta) (-i N^{-1/d}\nabla + A_{\epsilon})^2  + \delta^{-2} |A - A_{\epsilon}|^2,
\end{align*}
and the uniform upper bound on $\rho_{\gamma_N}/N$, the result follows for general $A$.

So we consider the Weyl asymptotics,
\begin{align*}
&\lim_{N\to\ii} -\frac1N\tr\big((-i N^{-1/d}\nabla + A)^2_{C_R}-c_{\rm TF}\rho(x)\big)_-\\
&\qquad\qquad=-\frac{1}{(2\pi)^d}\int_{\R^d}\int_{\R^d}\big(|p+A(x)|^2-c_{\rm TF}\rho(x)\big)_-\,dx\,dp\\
&\qquad\qquad=\inf_{0\leq m'\leq 1} \frac{1}{(2\pi)^d}\int_{\R^d}\int_{\R^d}\big(|p+A(x)|^2-c_{\rm TF}\rho(x)\big)m'(x,p)\,dx\,dp.
\end{align*}
with unique minimizer $m'(x,p) = \1(|p+A(x)|^2-c_{\rm TF}\rho(x)\leq0)$ in $L^\ii(\R^d\times\R^d)$, we conclude that $m=\1(|p+A(x)|^2-c_{\rm TF}\rho(x)\leq0)$ a.e. This gives in particular that 
$$N^{-1}\rho_{\gamma_N}(x)\wto \rho_m(x)=\frac{1}{(2\pi)^d}\int_{\R^d}\1(|p+A(x)|^2-c_{\rm TF}\rho(x)\big)\,dp=\rho(x)$$
weakly in $L^1\cap L^\ii$, hence that
\begin{equation}
N^{-1}\tr(\gamma_N)=N^{-1}\int_{C_R}\rho_{\gamma_N}\to\int_{C_R}\rho=1.
\label{eq:nb_particles}
\end{equation}
The latter says that $(-i N^{-1/d}\nabla + A)^2_{C_R}-c_{\rm TF}\rho$ has $N+o(N)$ negative eigenvalues.
From the above limits we also have as desired
\begin{align*}
\lim_{N\to\ii} \frac1{N}\tr(-i N^{-1/d}\nabla + A)^2_{C_R}\gamma_N&=\frac{1}{(2\pi)^d}\int_{\R^d}\int_{\R^d}|p|^2\1(|p|^2-c_{\rm TF}\rho(x)\big)\,dx\,dp\\
&=\frac{d}{d+2}4\pi^2\left(\frac{d}{|S^{d-1}|}\right)^{2/d}\int_{\R^d}\rho(x)^{1+2/d}\,dx.
\end{align*}

Finally, the results are all the same for an orthogonal projection $\tilde\gamma_N$ on $N$ first eigenfunctions of $(-i N^{-1/d}\nabla + A)^2_{C_R}-c_{\rm TF}\rho$ since 
$\tr|\gamma_N-\tilde\gamma_N|=o(N)$
by~\eqref{eq:nb_particles} and therefore $\|\rho_{\gamma_N}-\rho_{\tilde\gamma_N}\|_{L^1}=o(N)$. 
For $N$ large we have
$$\1\big((-i N^{-1/d}\nabla + A)^2_{C_R}\leq \rho(x)-\eps\big)\leq \tilde\gamma_N\leq \1\big((-i N^{-1/d}\nabla + A)^2_{C_R}\leq \rho(x)+\eps\big)$$
since, by the above arguments with $\rho$ replaced by $\rho\pm\eps$, 
$$\tr\1\big((-i N^{-1/d}\nabla + A)^2_{C_R}\leq \rho(x)\pm\eps\big)\sim (1\pm\eps R^d)N.$$
From the above estimates we conclude that $\rho_{\tilde\gamma_N}/N$ is bounded in $L^\ii$, and therefore $\rho_{\tilde\gamma_N}/N$ has the same weak limit as $\rho_{\gamma_N}$ in $L^1\cap L^\ii$. We also have
$$\tr((-i N^{-1/d}\nabla + A)^2_{C_R}-c_{\rm TF}\rho)(\gamma_N-\tilde\gamma_N)=o(N)$$
which implies that 
$$\tr(-i N^{-1/d}\nabla + A)^2_{C_R}(\gamma_N-\tilde\gamma_N)=o(N^{1+2/d})$$
and concludes the proof of Lemma~\ref{lem:semi-classics}.\qed


\end{document}